\documentclass[twoside,12pt]{article}

\usepackage{amsmath}
\usepackage{amssymb}
\usepackage{pxfonts}
\usepackage{graphicx}
\usepackage{eucal}
\usepackage{mathrsfs}
\usepackage{theorem}
\usepackage{pifont}
\usepackage{sectsty}
\usepackage{amscd}
\usepackage{color}
\usepackage{fancyhdr}
\usepackage{framed}
\usepackage[all]{xy}
\usepackage{mathbbol}
\usepackage{mdframed}
\usepackage{graphicx}

\usepackage[usenames,dvipsnames,svgnames,table]{xcolor}
\usepackage{mathptmx} 
\definecolor{shadecolor}{rgb}{0.9,0.9,0.9}

\usepackage{xcolor} % Required for specifying colors by name
\definecolor{ocre}{RGB}{243,102,25} 
\definecolor{darkocre}{RGB}{121,51,12} 
\definecolor{lightocre}{RGB}{255,150,37} 
\definecolor{verylightocre}{RGB}{255,204,50} 
\definecolor{soccerfield}{RGB}{107,142,35} 
\definecolor{lightgray}{RGB}{200,200,200} 
\definecolor{warmblue}{RGB}{51,102,153} 
\definecolor{lightwarmblue}{RGB}{105,141,198} 
\definecolor{sepia}{RGB}{112,66,20} 

\usepackage{tikz} % Required for drawing custom shapes
\usepackage{tkz-euclide}
\usetkzobj{all}
\usepackage{pgfplots}

\newcommand{\btkz}{\begin{tikzpicture}}
\newcommand{\etkz}{\end{tikzpicture}}

\usepackage{amsmath}
\usepackage{amssymb}
\usepackage{graphicx}
\usepackage{eucal}
\usepackage{mathrsfs}
\usepackage{pifont}
\usepackage{cjhebrew}

\newcommand{\calA}{{\mathcal A}}
\newcommand{\calB}{{\mathcal B}}

\newcommand{\calT}{{\mathcal T}}

\newcommand{\bbR}{{\mathbb R}}

\newcommand{\brk}[1]{\left(#1\right)}          % \brk{.}     => (.)
\newcommand{\deriv}[2]{\frac{d#1}{d#2}}

\newcommand{\pd}[2]{\frac{\partial#1}{\partial#2}}
\newcommand{\pdd}[2]{\frac{\partial^2#1}{\partial#2^2}}
\newcommand{\pddm}[3]{\frac{\partial^2#1}{\partial#2\,\partial#3}}

\newcommand{\secref}[1]{Section~\ref{#1}}

\newcommand{\defref}[1]{Definition~\ref{#1}}

\newcommand{\corrref}[1]{Corollary~\ref{#1}}

\newcommand{\beq}{\begin{equation}}
\newcommand{\eeq}{\end{equation}}
\newcommand{\bsplit}{\begin{split}}
\newcommand{\esplit}{\end{split}}
\newcommand{\baligned}{\begin{aligned}}
\newcommand{\ealigned}{\end{aligned}}

\newcounter{sect}

\providecommand{\R}{\bbR}

\newcommand{\textand}{\quad\text{ and }\quad}
\newcommand{\Textand}{\qquad\text{ and }\qquad}

\theoremheaderfont{\fontfamily{pzc}\bfseries\large}
\newtheorem{theorem}{Theorem}[section]
\newtheorem{lemma}{Lemma}[section]
\newtheorem{proposition}{Proposition}[section]

\newtheorem{corollary}{Corollary}[section]
\newtheorem{definition}{Definition}[section]

{\theorembodyfont{\rmfamily}

\newcommand{\specexercise}[1]{}

\newenvironment{proof}{{\flushleft \emph{Proof}:}}{\hfill\ding{110}}
\newenvironment{remark}{{\flushleft \fontfamily{pzc}\bfseries\large Remark:}}{}

\newcommand{\QQ}{\mathfrak{Q}}
\newcommand{\GG}{\mathcal{G}}

\newcommand{\G}{\mathrm{G}}
\newcommand{\B}{\mathcal{B}}
\renewcommand{\S}{\mathcal{S}}
\newcommand{\vp}{\varphi}
\newcommand{\e}{\varepsilon}
\newcommand{\id}{\operatorname{id}}

\newcommand{\LL}{\mathcal{L}}
\newcommand{\ks}{\kappa^*TS}

\newcommand{\Q}{\mathcal{Q}}

\newcommand{\ds}{\left.\frac{d}{ds}\right|_{s=0}}
\newcommand{\inc}{\,\lrcorner\,}
\newcommand{\hxi}{\xi^\S}
\newcommand{\xiQ}{\xi^\Q}

\DeclareMathOperator{\graf}{Graph}

\DeclareMathOperator{\divergence}{div}

\DeclareMathOperator{\Aff}{Aff}
\DeclareMathOperator{\Vol}{Vol}

\newcommand{\Emb}{\operatorname{Emb}}
\newcommand{\Hom}{\operatorname{Hom}}

\newcommand{\ext}{\operatorname{\mathbb{E}}}

\numberwithin{equation}{section}

\begin{document}

\title{Covariant Linearization of elasticity}
\author{Raz Kupferman \and Elihu Olami }


%\author{Raz Kupferman$ {}^* $, Elihu Olami$ {}^* $ }
%\address{Einstein Institute of Mathematics,
%The Hebrew University of Jerusalem
%Jerusalem, 9190401, Israel}

\maketitle
\abstract{In this paper we derive a general linearized theory for first-order continuum dynamics on manifolds with  particular application to incompatible elasticity. We adopt a global approach viewing the equations of motion as a $1$-form on the configuration space which is  the Banach manifold of  $C^1$ time-dependent embeddings of a body manifold $\B$ into a space manifold $\S$.  The linearization is done by differentiating the equations 1-form with respect to an affine connection which we construct and study extensively. We provide detailed coordinate computations for the linearized equations of a large class of problems in continuum dynamics on manifolds.}
\tableofcontents

%%%%%%%%%
%%%%%%%%%%
\section{introduction}

The derivation of a linear theory from a nonlinear theorem is a central theme in mathematics, with innumerable applications in the  various sciences. In the context of continuum mechanics, and notably  in the  theory of elasticity, the linear theories actually preceded the nonlinear theories (see Maugin \cite{Mau16}). In fact, the equations of linear elasticity are commonly derived directly from the balance laws (assuming small deformations) (Gurtin \cite{Gur73}), rather than as  approximation to the nonlinear theory.

%moreover, it's equations may be derived\footnote{add something about other derivations of linear elasticity, (balance laws, Gurtin (in marsden))} independently of the nonlinear theory when  considering small enough deformations. However it's most subtle derivation is as a linearization of the (true) nonlinear theory. 

Linear theories of elasticity play several key roles in the analysis of nonlinear theories: (i) they serve as an intermediate step for  proving the existence and the uniqueness of solutions for nonlinear theories, (ii) solutions of nonlinear problems can sometimes be obtained as limits of sequences of solutions of linearized problems, and (iii) they serve as a central tool in stability analysis \cite{MH83}. 

The linearization of nonlinear continuum theories is nowadays a standard, however, its current scope does not fully cover the wealth of systems of current interest. To a large extent, existing linear theories address systems that are geometrically Euclidean. From a mathematical perspective, the state-space in continuum mechanics can be described as the embeddings of a body into a space, both viewed as differentiable manifolds. 

For example, in a class of elastic systems dealing with residually-stressed bodies, the body manifold is viewed as a smooth manifold endowed with a Riemannian metric; the metric represents local equilibrium distances and angles between neighboring material elements. A configuration is an embedding of the body manifold into the ambient space, which is usually assumed Euclidean, although, non-Euclidean ambient spaces are of relevance even without recurring to relativistic theories \cite{KOS17a}. When the geometries of the body of the ambient space are incompatible, there is no notion of stress-free reference configuration, hence the very notion of small deformations is not naturally defined as it is when both body and space are assumed Euclidean. Incompatible elasticity is just one example in which complex geometries interact in a non-trivial way with mechanical laws and material properties.

%Linearization of classical (nonlinear) elasticity has  been extensively studied since ? see []. However the linearization of incompatible elasticity, though computed in specific cases (see e.g. [])  has been much less studied systematically. 
%
%Incompatible elasticity deals with  residually stressed bodies, which are elastic bodies admitting no stress free configurations in $\R^3$. See e.g. [] and [].  Mathematically, certain residually-stressed elastic bodies may be modeled as smooth manifolds endowed with a Riemannian metric. The metric represents local equilibrium distances and angles between neighboring material elements. A configuration is an embedding of the body manifold into the ambient Euclidean space. The elastic energy associated with a configuration is a measure of mismatch between the intrinsic metric of the body and its "actual" metric, the pullback of the Euclidean metric by the configuration. The property of being residually-stressed is a geometric incompatibility, reflected, in the traditional Euclidean settings, by the non-flatness of the intrinsic material metric. 
%
%In fact, incompatible elasticity can be viewed in the broader context elasticity between manifolds (see []) were the ambient space may bee curved as well. More generally, it can be viewed as a special case of continuum dynamics on  manifolds (see []). 

Physical theories in which non-Euclidean geometry plays a central role are best formulated in a covariant manner, i.e., in a way that does not rely on a particular system of coordinates. The classical reference for the covariant linearization of elasticity theories is the book of Marsden and Hughes \cite{MH83}. Their starting point is a general notion of linearization, which we hereby define:

%%%%%%%%%%%%
\begin{definition}
\label{def:linear}
Let $\pi:E\to M$ be a smooth (possibly infinite-dimensional) vector bundle endowed with a connection $\nabla$. Let $s\in\Gamma(E)$ be a $C^\infty$-section of $E$. The linearization of $s$ at $x\in M$  is an affine mapping $L_x(s):T_xM\to E_{s_x}$ given by 
\[
L_x(s)(w)=s_x + (\nabla_w s)_x,\quad \forall w\in T_xM.
\]
\end{definition}
%%%%%%%%%%

Marsden and Hughes formulate the equations of nonlinear elasticity as a section of an infinite-dimensional vector bundle over the manifold of configurations and compute their linearization for a general class of constitutive relations.  In their calculation, however, it is implicitly assumed that the ambient space is Euclidian, hence that the manifold of configurations is a vector space. This assumption is reflected in the linearization of the acceleration vector field and more subtly, in the linearization of the stress tensor. Accounting for a non-Euclidean ambient space is not just a matter of technicalities, which might be overcome, for example, by adopting a local coordinate system. A curved space affects the basic notion of inertia, and may destroy the symmetries that are at the heart of the classical derivation of continuum theories; this lack of symmetries reflects, for example, in the presence of so-called self-forces, which arise from interactions of the body with inhomogeneous geometric incompatibilities. 

%In the case where the ambient space is not homogeneous, \Raz{we need to be here very accurate}
%
%the stress depends not only on the configurations derivatives but also on their values, that is, on the specific location in the space manifold \Elik{(add ref)}. \Elik{not sure i'm right here...)}

Other approaches to covariant linearization can be found in Yavari and Ozakin \cite{YO08}, where the authors linearize the energy and momentum balance laws, and in \cite{GPM13} where linearization is computed around a normal state. 
%\Elik{Add more refs}. 
 
In this paper, we derive a general linearized theory for first-order continuum dynamics on manifolds, with a particular application to incompatible elasticity. We adopt a global approach, where the space of configurations $\QQ$ is the Banach manifold of $C^1$ time-dependent embeddings of a body manifold $\B$ into a space manifold $\S$. In this setting, the equations of motion are a 1-form on the configuration space $\QQ$. The linearization of those equation is in the sense of \defref{def:linear}, where the connection $\nabla^{\Q^*}$ on the cotangent bundle of the configuration space is induced in a natural way from a given connection $\nabla^\S$ on the tangent bundle of the space manifold. 

In the global approach to continuum dynamics,  the equations of motion can be viewed as a natural generalization to Newton's  laws. Velocity is the time derivative of the configuration; the acceleration is the covariant time-derivative of the velocity field with respect to the connection $\nabla^\QQ$; the force field, which is a 1-form $F\in\Omega^1(\QQ)$, is composed of external loadings and internal forces, where the latter are determined by the material properties through a constitutive relation. The equations of motion is obtained by pairing the acceleration to the force via a Riemannian metric $\GG$ on the configuration space $\QQ$.

Generally, elements of $T^*\Q$ are represented  by vector-valued measures. Hence, the linearized equations of motion may be as singular as measures and in particular, assume no local differential form. However, in the case where the loadings and the constituting relations satisfy certain regularity properties, the equations of motion as well as their linearization have local forms. which we derive as well.
% As a main application, we linearize the equations of motion for residually-stressed hyperelastic materials, assuming a prototypical constitutive relation. In this case, the local form of the equations of motion can be written in explicit form (for an arbitrary system of coordinates).  

%. We consider a quadratic hyperelastic constitutive model. We write the equations of motion in explicit form for the case of a free boundary, yielding a nonlinear wave equation and its corresponding linearized wave equation. As a particular system, we consider the case where B and S are two-dimensional, azimuthally symmetric annuli of different constant curvatures. We then present numerical calculations comparing the nonlinear and linear waves for the case of a spherical annulus embedded in a sphere of different radius.
 
The structure of the paper is as follows: In \secref{sec:prelims} we discuss the geometric structure of the space $C^1(M,N)$ where $M$ is a compact smooth manifold and $N$ is a smooth manifold without boundary. We first introduce the Banach manifold structure of $C^1(M,N)$ and its tangent bundle $TC^1(M,N)$. Next, we construct a metric and connection on $TC^1(M,N)$. To this end, we assume that a Riemannian metric $G$ is given on the target space $N$ and that a volume form  $\theta$ is prescribed for the source manifold $M$. The connection $\nabla^\QQ$ is induced by a connection $\nabla^N$ for $TN$. We discuss the construction of $\nabla^\Q$ in detail, and show that if $\nabla^N$ is metric with respect to $G$ then so is $\nabla^\Q$ with respect to $\GG$. 
 
In \secref{sec:Elastodynamics}, we use the results of \secref{sec:prelims} to formulate  Newton's equations for continuum dynamics. We identify the configuration space $\QQ$ of time-dependent $C^1$-embeddings as an open subset of the manifold $C^1(I\times\B,\S)$. The  connection for $T\QQ$ gives a notion of covariant derivative that defines the acceleration, whereas the metric for $T\QQ$ pairs the acceleration with force.  The force part of the equation is induced by a constitutive relation (which is assumed time-independent) and a loading;  the whole equation is viewed as a section of the cotangent bundle of the configuration space.
 
In \secref{sec:linearization}, we derive the linearized form of the nonlinear equations of motion derived in \secref{sec:Elastodynamics}. We first obtain a general expression for general, time-independent constitutive relations. We then derive a local differential representation for the case of a smooth constitutive relation; the linearized equations are formulated both in a covariant manner and in local coordinates.
 
%In section \ref{sec:incompatible_elasticity} we \Raz{complete after we complete this section}

 %%%%%%%%%%%%%
 %%%%%%%%%%%%%
 %%%%%%%%%%%%%%
 \section{Geometric preliminaries}
\label{sec:prelims}

In this section we present the geometric foundations for continuum dynamics on manifolds. We start by briefly recalling the notion of jets, which are the covariant constructs for encoding functions along with their derivatives.

\subsection{Jet bundles}
%%%%%%

\begin{definition}
Let $M$ and $N$ be smooth manifolds of dimensions $m$ and $n$. A \emph{$1$-jet} from $M$ to $N$ is an equivalence class of triples $(f,U,p)$, where $p\in M$, $U\subset M$ is a neighborhood of $p$ and $f\in C^1(U,N)$. Two triples $(f,U,p)$ and $(g,V,q)$ are equivalent if
\begin{enumerate}
\item $p=q$.
\item $f(p)=g(q)$.
\item There exists local charts in $M$ and $N$, with respect to which the local representatives of $f$ and $g$ have the same values and first derivatives at $p$.
\end{enumerate} 
Equivalently, $(f,U,p)$ and $(g,V,q)$ are equivalent if 
\[
(Tf)_p=(Tg)_q
\]
where $(Tf)_p:T_pM\to T_{f(p)}N$ is the tangent map of $f$ at $p$. 
We denote the $1$-jet of $f$ at $p$ by 
\[
[(f,U,p)]=j^1_pf.
\]
\end{definition}

\begin{remark}
The third condition in the definition of a 1-jet implies that $f$ and $g$ have the same values at $p$ and the same first derivatives at $p$ with respect to \emph{any} local coordinate charts. 
\end{remark}

We denote by $J^1(M,N)$ the set of all $1$-jets from $M$ to $N$. The set $J^1(M,N)$ can be given the structure of a smooth manifold of dimension $m+n+mn$; it is also a fiber bundle over $M$ with respect to the (source) projection map 
\[
\pi^1:J^1(M,N)\to M,\quad j^1_pf\mapsto p.
\]
Let $\pi:E\to M$ be a smooth vector bundle over $M$. Define 
\[
J^1(E)=\{j^1_ps ~:~ \text{$p\in M$ and $s$ is a local $C^1$-section of $E$ at $p$} \}.
\]
Then $\pi^1:J^1(E)\to M$ is a vector bundle over $M$. The first jet extension,
\[
j^1:C^1(E)\to C^0(J^1(E)),\quad s\mapsto j^1s,
\]
is a linear immersion. 

%%%%%%%%%%%%%%%%%%%%%%%%%%%%%%%%%%%
\subsection{The manifold $C^1(M,N)$}
\label{subsec:manifold}

Let $M$ be a smooth, compact, orientable $d$-dimensional manifold, and let $N$ be a smooth orientable $m$-dimensional manifold without boundary endowed with a Riemannian metric $G$. 
Let $C^1(M,N)$ be  the space of $C^1$ mappings $M\to N$. Endow $C^1(M,N)$ with the  Whitney $C^1$ -topology \cite{Mic80}, a subbase of which consists of sets of the form
\[
\{f\in C^1(M,N) ~:~ j^1f(M)\subset U    \},\quad U\subset J^1(M,N)\,\text{ is open}.
\]
Loosely speaking, the Whitney $C^1$-topology is the topology of uniform convergence of the function and its first derivative. 

The space $C^1(M,N)$ is not a vector space, since $N$ is not a linear space. However, $C^1(M,N)$ can be given a structure of an infinite-dimensional Banach manifold: a topological space locally homeomorphic to a Banach space and equipped with a smooth structure (see Lang \cite{Lan99}). 

Given a mapping $\kappa\in C^1(M,N)$, a coordinate chart for $C^1(M,N)$ at $\kappa$ is constructed as follows: Let $\nabla^N$ be the Levi-Civita connection of $G$ and let $\exp^N:D\to S$ be the  corresponding exponential map, where $D\subset TN$ is a neighborhood of the zero section of $TN$, such that 
\[
(\pi_N,\exp^N):D\to N\times N
\]
is an embedding (i.e., a diffeomorphism onto its image). Let 
\[
(\kappa^*\pi_N,\kappa^*\exp):\kappa^*D\to M\times N
\]  
be the embedding induced by the pullback with $\kappa$, and denote its image by $V_\kappa$. Then, the canonical chart at $\kappa$ 
\[
\phi_\kappa: C^1(\kappa^*D)\to U_\kappa,\qquad U_\kappa=\{f\in C^1(M,N) ~:~ \graf(f)\in V_\kappa\}
\]
is given by 
\beq
\label{canonicalchart}
\phi_\kappa(v)(p)=\exp(v_p).
\eeq
It's inverse $\phi_\kappa^{-1}:U_\kappa\to C^1(\kappa^*D)$ is given by
\[
\phi_\kappa^{-1}(f)=(\kappa^*\pi_N,\kappa^*\exp)^{-1}(Id,f).
\]
The differentiable structure obtained by the atlas 
\[
\{(\phi_\kappa,U_\kappa)~:~\kappa\in C^1(M,N)   \}
\] 
is independent of the choice of connection on $N$. 
For more detailed constructions see \cite{Eli67,Pal68,Mic80} and for alternative approaches see also \cite{PT01}. 

Since $D\subset TN$ is open, $C^1(\kappa^*D)\subset C^1(\kappa^*TN)$ is open. Since $C^1(M,N)$ is locally identified with $C^1(\kappa^*D)$, it follows that the tangent space $T_\kappa C^1(M,N)$ is  isomorphic to the Banachable space of vector fields along $\kappa$,
\[
C^1(\kappa^*TN)\simeq\{ v\in C^1(M,TN)~:~\pi_N\circ v=\kappa\}.
\]
The Banach space structure for $C^1(\kappa^*TN)$ may be constructed as follows: Let $||\cdot||:J^1(\kappa^*TN)\to\R$ be a Finsler structure on $J^1(\kappa^*TN)$, that is, for every $p\in M$, $||\cdot||_p:J^1(\kappa^*TN)_p\to\R$ is a norm and $||\cdot||$ varies smoothly between the fibers of $J^1(\kappa^*TN)$. Since $M$ is compact, a Finsler structure exists, and moreover, any two Finsler structures on $J^1(\ks)$ are equivalent. We define a complete norm on $C^1(\ks)$ by 
\[
||w||_{C^1}=\sup_{p\in M}||j^1_pw||.
\]
One may verify that the topology induced by the norm $||\cdot||_{C^1}$ on $C^1(\ks)$ coincides with its Whitney $C^1$ topology.  Thus, the canonical chart $\phi_\kappa$ is indeed a homeomorphisms onto its image.

The tangent bundle $TC^1(M,N)$ may be identified with the bundle 
\[
C^1(\pi_N):C^1(M,TN)\to C^1(M,N),
\]
where
\[
C^1(\pi_N)(w)=\pi_N\circ w.
\]
Moreover, for every $\kappa\in C^1(M,N)$  the mapping 
\[
\Phi_\kappa: C^1(\kappa^*D)\times C^1(\kappa^*TN)\to C^1(M,TN),
\]
given by 
\beq
\label{canonicaltrivialisation}
\Phi_\kappa(u,w)=T(\exp_{\pi_N(u)})_u(w),
\eeq
is a trivialisation for $C^1(M,TN)$ along the canonical chart $\phi_\kappa$ corresponding to the trivialisation $T(\phi_\kappa)$ for $TC^1(M,N)$ under the bundle equivalence. For details see Eliasson \cite{Eli67}. Note that for $(u,w)\in \kappa^*D\times_M \kappa^*TN$ 
\beq
\label{jacobitrivialisation}
T(\exp_{\pi_N(u)})_u(w)=J_{u,w}(1)
\eeq
where $J_{u,w}:[0,1]\to TN$ is the  unique Jacobi field along along the geodesic $\exp(tu)$ satisfying $J_{u,w}(0)=0$ and $\frac{DJ_{u,w}}{dt}(0)=w$. 

%%%%%%%%%%%%%%%%%%%%%%%%%%%%%%%%%%%%%%%
\subsection{Connection and metric for $C^1(M,N)$}

Following Eliasson \cite{Eli67}, we construct a connection for $TC^1(M,N)$. Let $\pi:E\to M$ be a (possibly infinite dimensional) fiber bundle over a smooth manifold $M$ and let $VE \subset TE\to E$ be the vertical bundle defined by $VE=\ker(d\pi)$. An Ehresmann connection is a splitting $\tilde{K}:TE\to VE$ of the short exact sequence
\[
\begin{xy}
(0,0)*+{0}="zero1";
(18,0)*+{VE}="VE";
(36,0)*+{TE}="TE";
(54,0)*+{\pi^*TM}="TM";
(72,0)*+{0}="zero2";
{\ar@{->}"zero1";"VE"};
{\ar@{->}"TM";"zero2"};
{\ar@{->}^{\iota_V}"VE";"TE"};
{\ar@{->}^{d\pi}"TE";"TM"};
{\ar@{-->}@/^{1pc}/^{\mathcal{K}}"TE";"VE"};
\end{xy}
\]
satisfying  $\tilde{K}\circ \iota_V=\id_{VE}$ where $\iota_V:VE\to TE$ is the inclusion. $\mathcal{K}$ is often referred to as the \emph{connection form} of the Ehresmann connection. The horizontal bundle $HE$ is then identified with $\ker(\tilde{K})$. In case that $\pi:E\to M$ is a vector bundle we have a canonical identification $VE\simeq \pi^*E$. Thus, $\tilde{K}$ induces a unique mapping $K:TE\to E$ which we call a \emph{connection map} for $E$.
\[
\begin{xy}
(0,0)*+{E}="E1";
(20,0)*+{M}="M";
(20,20)*+{E}="E2";
(0,20)*+{\pi^*E}="piE";
(-20,30)*+{TE}="TE";
{\ar@{->}_{\pi}"E1";"M"};
{\ar@{->}^{\pi}"E2";"M"};
{\ar@{->}_{\pi^*(\pi)}"piE";"E1"};
{\ar@{->}^{\tilde{\pi}}"piE";"E2"};
{\ar@{->}^{\tilde{K}}"TE";"piE"};
{\ar@{->}@/_{1pc}/_{\pi_E}"TE";"E1"};
{\ar@{->}@/^{1pc}/^{K}"TE";"E2"};
\end{xy}
\]
A linear connection should also satisfy, the following condition: for every $\lambda\in \R$ denote by $S_\lambda:E\to E$  scalar multiplication by $\lambda$, then for every $e\in E$,
\[
K\circ(dS_\lambda)_e=\lambda K_e .
\]
 Suppose that $M$ and $E$ are modelled over the Banach spaces $\widetilde{M}$ and $\hat{E}$ respectively. Then $K$  has the local form
 \[
\tilde{K}:\widetilde{M}\times \hat{E}\times\widetilde{M}\times\hat{E}\to \widetilde{M}\times\hat{E}
\]
\[
\tilde{K}(x,\xi,\tilde{x},\tilde{\xi})=(x,S(x,\xi)(\tilde{x},\tilde{\xi})),
\]
where $S(x,v)(\tilde{x},\tilde{\xi})$ is linear in $\tilde{x}$ and $\tilde{\xi}$. The condition $\mathcal{K}\circ \iota_V=\id$ implies that $S$ is of the form $S(x,v)(\tilde{x},\tilde{\xi})=\tilde{\xi}-\Gamma(x,\xi)(\tilde{x})$ and the linearity condition implies that $\Gamma$ is linear in $\xi$. Thus, a linear connection map has the local form 
\[
\tilde{K}(x,\xi,\tilde{x},\tilde{\xi})=(x,\tilde{\xi}-\Gamma(x)(\xi,\tilde{x})),
\]
where $\Gamma(x):\widetilde{M}\times \hat{E}\to \hat{E}$ is a bilinear transformation called the \emph{local connector} of $K$ at $x$. 

In the particular case where $M$ is finite-dimensional and $E=TM$, the local connector $\Gamma$ is given by the Christoffel symbols, 
\[
\Gamma(x)(v^ie_i,w^je_j)=\Gamma^k_{ij}(x)v^iw^j \, e_k.
\] 

Given a connection map $K$ for $E$, one can define a covariant derivative $\nabla$ on $E$ in the following way:
For a section $\xi\in\Gamma(E)$, set its covariant derivative as $\nabla\xi=K\circ T\xi\in\Gamma(\Hom(TM,E))$. That is, for $p\in M$ and $w\in T_pM$ 
\[
(\nabla_w\xi)_p=K(T\xi(w))\in E_p.
\]
If a section $\xi$ is represented by $\tilde{\xi}:\tilde{M}\to \hat{E}$, that is, locally $\xi(\cdot)=(\cdot,\tilde{\xi}(\cdot))$,  and $w\in T_pM$ has a local representation $(x,\tilde{w})$, then a simple computation gives that the coordinate representation of $(\nabla_w\xi)_p$ is
\[
(x,D\tilde{\xi}(x)(\tilde{w})+\Gamma(x)(\tilde{w},\tilde{\xi}(x))),
\]
where $x$ is the coordinate corresponding to $p$.
\\
% Reciprocally,
%a connection $\nabla$ for $E$ induces a connection map $K:TE\to E$ in the following way: \Elik{To Complete!}
 
Turning back to the problem at hand, let $E=TN$ and let $K^N:T^2N\to TN$ be the connection map corresponding to the Levi Civita connection $\nabla^N$ on $TN$. One can then show (see \cite{Eli67} for details) that $K^N$ induces a connection map  
\[
C^1(K^N):T^2C^1(M,N)\simeq C^1(M,T^2N)\to  TC^1(M,N)\simeq C^1(M,TN)
\]
%for $C^1(M,TN)\simeq TC^1(M,N)$, 
defined by composition,
\beq\label{connectionmapformula}
C^1(K^N)(A)=K^N\circ A,\quad A\in C^1(M,T^2N).
\eeq
Denote the  corresponding connection by $\nabla^{C^1}$.
By definition, for $\xi\in\Gamma(TC^1(M,N))$, $\kappa\in C^1(M,N)$ and $w\in T_\kappa C^1(M,N)$,
\beq
\label{conformula}
(\nabla^{C^1}_w\xi)_\kappa=(C^1(K^N)\circ (T\xi)_\kappa)(w)=K^N \circ\left((T\xi)_\kappa(w) \right).
\eeq
Note that on the right-hand side, $(T\xi)_\kappa(w) :M\to T^2N$ and $K^N:T^2N\to TN$, hence, we obtain indeed a map $M\to TN$, i.e., an element of $TC^1(M,N)$.
The exponential map for $C^1(M,N)$ with respect to $\nabla^{C^1}$ is given by composition with $\exp$, thus, the canonical coordinate charts $\phi_\kappa$ are normal coordinates in the following sense: 
for every $\xi\in T_\kappa C^1(M,N)$, $t\mapsto \phi_\kappa(t\xi)$ is a $\nabla^{C^1}$-geodesic \cite{Eli67}. In particular, the local connector of $\nabla^{C^1}$ in the canonical coordinate chart $\phi_\kappa$ vanishes at the zero section $\bar{0}\in C^1(\kappa^*TN)$ (corresponding to $\kappa$). 

We next turn to construct a Riemannian metric for $C^1(M,N)$. Assume that a \emph{mass form}, which is a positive $d$-form $\theta$ on $M$  is given. Using the isomorphism 
\[
TC^1(M,N)\simeq C^1(M,TN),
\]
define a metric $\GG$ for $C^1(M,N)$ by 
\beq
\label{metricdefinition}
\GG_\kappa(u,w)=\int_M \kappa^*G(u,w)\theta,\quad u,w\in T_\kappa C^1(M,N) \simeq C^1(\kappa^*TN).
\eeq
The mass density of $M$ is incorporated in the mass form $\theta$. Locally,
\[
\theta=\rho \, dx^1\wedge\cdots\wedge dx^d,
\]
where $\rho:M\to \R_+$ is a mass density function. In cases where $M$ is endowed with a Reimannian metric $g$, it is often natural to take for mass form the Riemannian volume form $\theta=\Vol_g$, corresponding to the mass destiny $\rho=\sqrt{\det(g_{ij})}$ . 

\begin{remark}
As always, the metric $\GG$ induces an isometric immersion $\flat^{\GG}:TC^1(M,N)\to T^*C^1(M,N)$ given by 
\[
\flat^\GG(w)=\GG(w,\cdot),\quad \forall w\in TC^(M,N).
\]
However, since the manifold is not a Hilbert manifold, $\flat^\GG$ is not an isomorphism. For this reason, $\GG$ is often called a \emph{weak Riemannian structure} (as opposed to a strong Riemannian structure). 
\end{remark}

\subsection{Metricity of the connection}

We next show that the connection $\nabla^{C^1}$ and the metric $\GG$ for $C^1(M,N)$ are compatible, namely, for $u,v,w\in \Gamma(TC^1(M,N))$,
\[
u (\GG(v,w)) = \GG(\nabla^{C^1}_uv,w) + \GG(v,\nabla^{C^1}_uw).
\]
The metricity of the connection will be used in several instances in the mechanical context.

\begin{lemma}
\label{jacobilemma1}
Let $y\in N$, $v,w\in T_yN$ and $t>0$ sufficiently small. Then, for every $s\in[0,1]$,
\[
J_{tv,w}(s)=J_{v,w/t}(ts),
\]
where $J_{v,w}$ is the Jacobi field as defined in \secref{subsec:manifold}.
\end{lemma}
\begin{proof}
Define $J_1,J_2:I\to TN$ by $J_1(s)=J_{tv,w}(s)$ and $J_2(s)= J_{v,w/t}(ts)$. 
We need to prove that $J_1=J_2$.
$J_1$ is a Jacobi field along the geodesic $\gamma_t(s)=\exp(stv)$.  Since $\dot{\gamma}_t(s)=t\dot{\gamma}_1(ts)$ and $J_{v,w}$ satisfies the Jacobi equation, we get that
\[
\begin{split}
\left.\frac{D^2}{ds^2}J_2\right|_s&=t^2\left.\frac{D^2}{ds^2}J_{v,w/t}\right|_{ts}=-t^2R(J_{v,w/t}(ts),\dot{\gamma}_1(ts)),\dot{\gamma_1}(ts)=\\
&=-R(J_2(s),\dot{\gamma}_t(s)), \dot{\gamma}_t(s).
\end{split}
\]
In other words, $J_2$ is also a Jacobi field along $\gamma_t$. Moerover,  $J_1(0)=J_2(0)=0$ and  
\[
\frac{D}{ds} J_2(0)=t\,\frac{D}{ds} J_{v,w/t}(0)= t\cdot \frac{w}{t}=w=
\frac{D}{ds} J_1(0).
\]
The result  follows from the existence and uniqueness of solutions to ordinary differential equations.
\end{proof}

The following lemma is a standard result in the theory of Jacobi fields (see e.g. [DOC92]).
\begin{lemma}\label{jacobilemma2}
Let $y\in N$, and $(y^i)$ be normal coordinates centered at $y$, and let $\gamma(t)=\exp(tv)$ ($v\in T_yN$) be a radial geodesic emanating from $y$. Then for any $w\in T_yN$ given locally by $w=w^i\partial_{y_i}$, the Jacobi field along $\gamma$ with initial conditions $J(0)=0$, $(D/dt)J(0)=w$ is given locally by 
\[
J(t)=tw^i(\gamma^*\partial_{y_i}).
\]
\end{lemma}
%\Elik{The above lemma is a well known result sometimes taught in basic geometry courses, perhaps it should be omitted. On the other hand, lemma \ref{jacobilemma1} follows easily from lemma \ref{jacobilemma2}...}
\begin{theorem}
\label{thm:metricity}
The connection $\nabla^{C^1}$ is metric with respect to $\GG$. In other words, for every $u,v,w\in \Gamma(TC^1(M,N))$ 
\[
(\nabla^{C^1}_u\GG)(v,w) := u(\GG(v,w))-\GG(\nabla^\Q_uv,w)-\GG(v,\nabla^\Q_uw)=0.
\]
\end{theorem}
\begin{proof}
Let $\kappa\in C^1(M,N)$. It suffices to show that  $\nabla^{C^1}\GG$ vanishes at some coordinate chart at $\kappa$. Let
\[
\phi_\kappa: C^1(\kappa^*D)\to C^1(M,N).
\]
be the canonical chart around $\kappa$ and let
\[
\Phi_\kappa: C^1(\kappa^*D)\times C^1(\kappa^*TN)\to C^1(M,TN)\simeq TC^1(M,N)
\]
be the corresponding trivialization of $TC^1(M,N)$ along $\phi_\kappa$ given by \eqref{canonicalchart} and \eqref{canonicaltrivialisation}. Since $\phi_\kappa$ is a normal coordinate chart, the Christoffel symbols (i.e., the local connector) of $\nabla^\Q$ vanish at $\kappa$. Therefore, it suffices to prove that the derivative of the local representative of $\GG$ vanishes at the zero section $\bar{0}\in C^1(\kappa^*TN)$ (corresponding to $\kappa$). 

The local representative of $\GG$,
\[
\tilde{\GG}:C^1(\kappa^*TN)^*\otimes C^1(\kappa^*TN)^*,
\]
is given by 
\[
\tilde{\GG}(\xi_0)(\xi_1,\xi_2)=\GG_{\phi_\kappa(\xi_0)}(\Phi_\kappa(\xi_0,\xi_1),\Phi_\kappa(\xi_0,\xi_2)),
\]
where $\xi_0\in C^1(\kappa^*D)$ and $\xi_1,\xi_2\in C^1(\kappa^*TN)$. 
More explicitly, using \eqref{jacobitrivialisation},
\[
\begin{split}
\tilde{\GG}(\xi_0)(\xi_1,\xi_2)&=\int_M G_{\exp(\xi_0)}(T(\exp_{\pi_N(\xi_0)})_{\xi_0}(\xi_1),T(\exp_{\pi_N(\xi_0)})_{\xi_0}(\xi_2))\theta \\
&=\int_M G_{\exp(\xi_0)}(J_{\xi_0,\xi_1}(1),J_{\xi_0,\xi_2}(1))\theta.
\end{split}
\]
Note that the vector field $J_{\xi_0,\xi_1}(1)$ evaluated at $p\in M$ is given by $J_{\xi_0(p),\xi_1(p)}(1)$.

Now, Let $\xi_0=\bar{0}\in C^1(\kappa^*D)$ (so that $\phi_\kappa(\xi_0)=\kappa$), $\eta\in C^1(\kappa^*TN)$ and $\xi_1,\xi_2$ as before. Then 
\[
\begin{split}
D&\tilde{\GG}_{\bar{0}}(\eta)(\xi_1,\xi_2)=\left.\frac{d}{dt}\right|_{t=0}\bar{\GG}(t\eta)(\xi_1,\xi_2)\\
&=\left.\frac{d}{dt}\right|_{t=0}\int_M G_{\exp(t\eta)}(J_{t\eta,\xi_1}(1),J_{t\eta,\xi_2})\theta\\
&=\int_M \left.\frac{d}{dt}\right|_{t=0}G_{\exp(t\eta)}(J_{t\eta,\xi_1}(1),J_{t\eta,\xi_2})\theta\\
&=\int_M \left.\frac{d}{dt}\right|_{t=0}G_{\exp(t\eta)}(J_{\eta,\xi_1/t}(t),J_{\eta,\xi_2/t}(t))\theta,
\end{split}
\]
where in the passage to the third line we interchange integration over $M$ and differentiation with respect to time, and the last equality follows from lemma \ref{jacobilemma1}. 

It suffices to show the the integrand vanishes at every $p\in M$. Let $p\in M$ and $v,u,w\in T_{\kappa(p)}N$  we need to prove that 
\beq\label{eq:1}
\lim_{t\to 0}\frac{d}{dt}G_{\exp(tv)}(J_{v,u/t}(t),J_{v,w/t}(t))=0.
\eeq
Since $\nabla^N$ is metric with respect to $\G$,
\[
\begin{split}
\frac{d}{dt}G_{\exp(tv)}&(J_{v,u/t}(t),J_{v,w/t}(t))=G_{\exp(tv)}\brk{\frac{D}{dt}J_{v,u/t}(t),J_{v,w/t}(t)} \\
&+G_{\exp(tv)}\brk{J_{v,u/t}(t),\frac{D}{dt}J_{v,w/t}(t)}.
\end{split}
\]
Let $(y^i)$ be normal coordinates centred at $\kappa(p)$, and let $\gamma(s)=\exp(sv)$. Then by lemma \ref{jacobilemma2}, $J_{v,w/t}$ is given locally by $J_{v,w/t}(s)=sw^i/t(\gamma^*\partial_{y^i})$ hence, $J_{v,w/t}(t)=w^i\partial_{y^i}|_{\gamma(t)}$ is a constant vector field along $\gamma$ and 
\[
\lim_{t\to 0}\frac{D}{dt}J_{v,w/t}(t))=0,
\]
which completes the proof.

\end{proof}
\begin{remark}
The proof of theorem \ref{thm:metricity} shows in fact, that $\nabla^{C^1}$ is metric with respect to $\GG$  whenever $\nabla^N$ is metric with respect to $G$. Note that metricity does not depend on the choice of mass form $\theta$ for $\B$. 
\end{remark}
%%%%%%
%%%%%%%
%%%%%%%
%%%%%%%%

%%%%%%%%%%%%%%%%%%%%%%%%%%%%%%%%%%%%%%%%%%%%%%%%%%%%
\section{Elastodynamics}
\label{sec:Elastodynamics}

In this section we give a brief review of the geometric setting of elastodynamics. 
The exposition, which builds upon the geometric construction in \secref{sec:prelims}, follows the lines of \cite{KOS17a}.

%%%%%%%%%%%%%%%%%%%%%%%%%%%%%%%%%%%%%%%%%%%%%%%%%%%%
\subsection{The manifold of configurations}

%%%%%%%%%%%
\begin{definition}
A \emph{body manifold} $\B$ is a smooth compact and orientable $d$-dimensional manifold. A \emph{space manifold} $\S$ is a smooth orientable $m$-dimensional manifold without boundary. 
\end{definition}
%%%%%%%%%%%

We assume  that $\S$ is equipped with a Riemannian metric $G$ and that $\B$ is equipped with a mass form $\theta\in\Omega^d(\B)$. The canonical charts for $C^1(\B,\S)$ are constructed as in \secref{sec:prelims} using the exponential map induced by the Levi-Civita connection $\nabla^\S$ of $\S$.

%%%%%%%%%%%
\begin{definition}
Denote by 
\[
\Q = \Emb^1(\B,\S)
\]½
 the space of $C^1$-embeddings of $\B$ in $\S$. Let $I\subset\R$ be a closed time interval. The \emph{configuration space},
\[
\QQ = C^1(I,\Q),
\]
is the space of $C^1$-paths of embeddings of $\B$ in $\S$.
\end{definition}
%%%%%%%%%%%

Since $\Q$ is an open subset of $C^1(\B,\S)$ with respect to the Whitney $C^1$-topology (see  \cite{Mic80}), it inherits the Banach manifold structure of $C^1(\B,\S)$. Moreover, as (see \cite{Eli67})
\[
C^1(I\times\B,\S)\simeq C^1(I,C^1(\B,\S)),
\]
we may view $\QQ$ as an open subset of $C^1(I\times\B,\S)$. $\QQ$ therefore inherits the Banach manifold structure of $C^1(I\times \B,\S)$.
 
Note that there is a natural inclusion $\iota_\Q:\Q\hookrightarrow\QQ$, given by
\beq
(\iota_\Q\kappa)(t,p) = \kappa(p).
\label{eq:iotaQ}
\eeq
We refer to $\Q$ as the space of stationary configurations.  
 
%\Elik{Add some diagrams of mappings induced by $\iota_\Q$ , $(T\iota_\Q)^*:T^*\QQ\to T^*\Q$ and more. Perhaps later? }
 
The tangent bundle $T\QQ$ is called the bundle of \emph{virtual displacements}, or \emph{generalised velocities}. For $\kappa\in\QQ$, an element $v\in T_\kappa\QQ$ is called a \emph{virtual displacement} at $\kappa$. As in the general case, we have the isomorphisms,
\[
T_\kappa\QQ\simeq C^1(\ks)\simeq \{v\in C^1(I\times\B,T\S)~:~\pi_\S\circ v=\kappa\},
\]
and 
\[
T\QQ\simeq \{v\in C^1(I\times\B,T\S)~:~\pi_\S\circ v\in \QQ\}\subset C^1(I\times\B,T\S).
\]
where the above inclusion is open; in other words, we view $T\QQ$ as an open submanifold of $C^1(I\times\B,T\S)$.

%%%%%%%%%%%%
%\begin{remark}
%The isomorphism $T_\kappa\QQ\simeq C^1(\ks)$ provides a physical interpretation for virtual displacements. When considering the mechanics of a single particle, a virtual displacement assigns at every time a tangent vector at the current location, directing to a possible displacement. When considering the motion of a smooth body $\B$ in an ambient space $\S$, a virtual displacement at a configuration $\kappa$ assigns a velocity at $\kappa(t,p)$ at every time $t\in I$ and every body point $p\in\B$. 
%%In other words, a virtual displacement at $\kappa$ is a vector field along $\kappa$, $v\in C^1(\ks)$.
%\end{remark}
%%%%%%%%%%%%

Denote the restriction of the connection map $C^1(K^\S)$ (see \eqref{connectionmapformula}) to $\QQ$  by $K^\QQ$, 
that is,
\[
K^\QQ : T^2\QQ \subset C^1(I\times \B, T^2\S) \to T\QQ \subset C^1(I\times\B,T\S).
\]
Denote the corresponding connection by $\nabla^\QQ$, namely,
\[
\nabla^\QQ : \Gamma(T\QQ)\times \Gamma(T\QQ) \to \Gamma(T\QQ).
\] 
The metric $\GG$ for $\QQ$ is given by 
\beq
\label{metricforQ}
\GG_\kappa(v,w)=\int_{I\times B}\kappa^*G(v,w)\,\theta\wedge dt,\quad v,w\in T_\kappa\QQ\simeq C^1(\ks).
\eeq
By Theorem~\ref{thm:metricity}, $\nabla^\QQ$ is metrically consistent with $\GG$.

%The tangent bundle $T\QQ$ is called the bundle of \emph{virtual displacements}, or \emph{generalised velocities}. Let $\kappa\in\QQ$, an element $v\in T_\kappa\Q\simeq C^1(\ks)$ is called a \emph{virtual displacement} at $\kappa$. 
%\begin{remark}
%The isomorphism $T_\kappa\QQ\simeq C^1(\ks)$ provides a physical interpretation for virtual displacements. When considering the mechanics of a single particle, a virtual displacement assigns for every time a tangent vector at the current location  pointing  to a possible displacement direction. When considering the motion of the (continuous) body $\B$ in the ambient space $\S$ a virtual displacement at a configuration $\kappa$ should assign a direction at $\kappa(x)$ for every body point $x\in\B$. In other words, a virtual displacement at $\kappa$ is a vector field along $\kappa$, $v\in C^1(\ks)$.
%\end{remark}

Throughout this paper, points in $I\times\B$ and $\S$ are denoted by $(t,x)$ and $y$ respectively. The indices of coordinates in $I\times\B$ will be denoted by Greek letters, whereas indices of coordinates in $\S$ will be denoted by Roman letters. A point $(t,x)\in I\times \B$ is represented by $(x^\alpha)_{\alpha=0}^d=(t,x^1,\ldots,x^d)$.

%%%%%%%%%
\subsection{Forces and stresses}

\begin{definition}
Let $\kappa\in \QQ$. A  \emph{force at $\kappa$} is an element $f\in T_\kappa^*\QQ\simeq (C^1(\ks))^*$. The action $f(w)$ of a force $f\in T^*_\kappa\QQ$ on a virtual displacement  $w\in T_\kappa\QQ$ is called a \emph{virtual power}. 
\end{definition}

For simplicity, we will focus our attention on forces that are independent of time derivatives; that is, forces $f\in T_\kappa^*\QQ$ of the form 
\beq
\label{timeindforce}
f(w)=\int_I f_t(w_t)dt,\quad \forall w\in T_\kappa\QQ\simeq C^1(\ks),
\eeq
where$\{f_t\}_{t\in I}$ is a smooth family of elements $f_t\in T^*_{\kappa_t}\Q$, $\kappa_t = \kappa(t,\cdot)\in \Q$ and $w_t:=w_{(t,\cdot)}\in T_{\kappa_t}\Q\simeq C^1(\kappa_t^*T\S)$.  With a slight abuse of terminology, we refer to elements of $T^*\Q$ as forces as well.

We therefore turn to present the structure of $T^*\Q$, the space of forces over stationary configurations. 
First, note that unlike in finite dimensions, the tangent and cotangent bundles $T\Q$ and $T^*\Q$ are not isomorphic. In particular, given a stationary configuration $\vp\in\Q$,  the dual space $T^*_\vp\Q\simeq(C^1(\vp^*T\S))^*$ depends on the topology of $C^1(\vp^*T\S)$. Since the  topology of $C^1(\vp^*T\S)$ takes into account  first derivatives, so do the elements of $(C^1(\vp^*T\S))^*$. 
% In fact, this \Raz{(what is?)} is the mathematical origin of  stresses and their basic properties:
 
More formally, let $\vp\in\Q$, and consider the first jet extension
\[
j^1:C^1(\vp^*T\S)\to C^0(J^1(\vp^*T\S)),\qquad v\mapsto j^1v,
\]
 which is a continuous linear embedding. By the Hahn-Banach theorem, its dual map,
 \[
 (j^1)^*:(C^0(J^1(\vp^*T\S)))^*\to (C^1(\vp^*T\S))^*,
 \] 
is onto. We conclude that to every force $f$ at $\vp$ corresponds a (non-unique) $\sigma\in (C^0(J^1(\vp^*T\S)))^*$, satisfying
\beq
\label{stress-rep}
f(w)=(j^1)^*\sigma(w):=\sigma(j^1w),\quad \forall w\in C^1(\vp^*T\S).
\eeq
We call $\sigma$ a \emph{stress} at $\vp$. We say that a stress $\sigma$ at $\vp$ represents the force $f$ if the relation \eqref{stress-rep} holds. Note however, that for a given force $f$, there may be more than one stress representing it. This reflects the well-known static indeterminacy of continuum mechanics. 

In fact, stresses may  also be viewed as cotangent vectors of some other manifold; Let $\mathcal{E}=C^0(J^1(\B,\S))$ be the manifold of $C^0$-sections $\B\to J^1(\B,\S)$. Then for every $\vp\in\Q$ one has a canonical isomorphism 
\beq
(C^0(J^1(\vp^*T\S))^* \simeq T^*_{j^1\vp}\mathcal{E}.
\label{eq:iso5}
\eeq
For more details see \cite{KOS17a}. 

In general, stresses and forces, which are continuous linear functionals on differentiable sections, may be singular. Locally, and in particular, if $\B$ can be covered by a single chart, every stress $\sigma$ is represented by a collection of measures on $\B$, 
\[
\{\mu_i,\mu_i^\alpha~:~1\leq\alpha\leq d,\,1\leq i\leq m   \}
\]
by the formula 
\[
\sigma(j^1w)=\int_{\B}w^id\mu_i+\int_{\B}w^i_{,\alpha}d\mu_i^\alpha.
\]
If the measures  $\{\mu_i,\mu_i^\alpha\}$ are absolutely continuous, we may write 
\[
\mu_i=R_i \,\Vol
\Textand 
\mu_i^\alpha=S_i^\alpha\,\Vol ,
\]
where $R_i,S_i^\alpha\in C^1(\B)$ and $Vol:=dx^1\wedge\cdots\wedge dx^d$. This suggests the following definition (see \cite{Seg86}):

\begin{definition}
Let $\vp\in\Q$. A \emph{variational stress density} $S$ at $\vp$ is a smooth $d$-form valued in the vector bundle $(J^1(\vp^*T\S))^*$. In other words 
\[
S\in \Gamma(\Hom(J^1(\vp^*T\S),\Lambda^dT^*\B)).
\] 
 We say that a stress $\sigma$ at $\vp$ is \emph{smooth}, if there exists a variational stress density $S\in\Gamma(\Hom(J^1(\vp^*T\S),\Lambda^dT^*\B)$, such that 
\[
\sigma(j^1v)=\int_{\B}S(j^1v)
\]
for every $v\in C^1(\vp^*T\S)$.
\end{definition}
% 
%The restriction condition \eqref{restrictioncond} for $S$ means physically that the stress acts only on spacial derivatives of displacements. 
%\Raz{(Where does this come from?)}

%\begin{remark}
% \Elik{Something about how to view $C^0(J^1(\ks))^*$ as a cotangent space. $\s\simeq T^*_{j^1\kappa}C^0(J^1(\B,\S))$ add ref to stress-theory paper.} \Raz{(Indeed...)}
% \end{remark}
 
Let $S$ be a variational stress density at $\vp$.  As shown in \cite{Seg02,Seg13}, we may decompose $S$ into body and surface terms as follows,
\beq
\label{stressrep}
\int_{\B} S(j^1w)= - \int_{\B} \divergence S(w)+\int_{\partial\B} p_\sigma S(w).
\eeq
Here, $\divergence S$ and $p_\sigma S$ are vector-valued forms,
\[
\begin{gathered}
\divergence S \in \Gamma(\Hom(\vp^*T\S,\Lambda^{d}T^*\B)) \\
p_\sigma S \in \Gamma(\Hom(\vp^*T\S,\Lambda^{d-1}T^*\B) ),
\end{gathered}
\]
%Note that $\partial(I\times\B)=\partial I\times\B\cup I\times \partial\B$, condition [] implies that $p_\sigma(S)$ vanishes on $\partial I\times \B$. (\Elik{To  verify}).
%\Raz{(The question is whether we didn't create a monster and whether we really need to work here in $d+1$ dimension.)}

In coordinates, the action of a variational stress on the jet extension of a virtual velocity is of form 
\[
S(j^1w) = (R_i w^i + S_i^\alpha w^i_{,\alpha})\, \Vol, 
\]
where $R_i,S_i^\alpha\in C^1(\B)$. The vector-valued forms $\divergence S$ and $p_\sigma S$ are then given by
\[
\begin{split}
\divergence S(w) &= (\divergence S)_i w^i\, \Vol \\
p_\sigma S(w) &= (p_\sigma S)^\alpha_i w^i\,  \partial_\alpha\inc\Vol ,
\end{split}
\]
where
\beq
\label{eq:divS}
(\divergence S)_i = S^\alpha_{i,\alpha} - R_i
\Textand 
(p_\sigma S)_i^\alpha =  S^\alpha_i.
\eeq
%Here, the notation $\widehat{dx^\alpha}$ indicates that the term $dx^\alpha$ is omitted from the wedge product; in the expression $(-1)^{\alpha-1}S^\alpha_i$ there is no summation over $\alpha$. 
%Note that condition [] on $S$ implies that $(p_\sigma S)_i^t=0$ which is consistent with  $P_\sigma(S)|_{\partial I\times\B}\equiv 0$.  

Let $\vp\in\Q$. Suppose that a force $f\in T^*_\vp\Q$ is given by body and surface force densities $b\in \Gamma(\Hom(\vp^*T\S,\Lambda^{d}T^*\B))$ and $\mathcal{T}\in\Gamma(\Hom(\vp^*T\S|_{\partial\B},\Lambda^{d-1}T^*\partial\B))$. That is, for every $w\in C^1(\vp^*T\S)$
\beq
\label{smoothforce}
f(w)=\int_{\B} b(w) +\int_{\partial\B} \mathcal{T}(w),
\eeq
Then, it follows from \eqref{stressrep} that $f$ is represented by a smooth stress $\sigma$ at $\vp$ with variational stress density $S$,
\[
f(w) = \int_{\B} S(j^1 w),
\]
if and only if, for every virtual displacement $w\in C^1(\vp^*T\S)$,
\[
\int_{\B}\left( \divergence S(w) + b(w)\right)=0 
\]
and
\[
\int_{\partial\B}\left(p_\sigma S|_{\partial\B}(w)-\mathcal{T}(w)\right)=0.
\]
We conclude that $f$ is represented by a variational stress density $S$, if and only if 
\beq
\label{stressrep1}
\divergence S + b=0 
\Textand  
p_\sigma S|_{\partial\B}=\mathcal{T}.
\eeq
%Here, $p_\sigma S|_{\partial\B}=\iota_{\partial\B}^*p_\sigma S$ is the restriction of forms, where $=\iota_{\partial\B}:\partial\B\to \B$ is the inclusion.

Equation \eqref{stressrep1} is only a representation theorem. In other words, for every fixed $\vp\in\Q$ and force $f\in T_\vp^*\Q$ of the form \eqref{smoothforce}, a smooth stress $\sigma$, given by a variational stress density $S$ at $\vp$, represents $f$ if and only if $S$ satisfies the boundary value problem \eqref{stressrep1}. 
 
Note also that equation \eqref{stressrep1} is underdetermined: in local charts it constitutes $d$ equations for the $(d\times m+m)$ components $(R_i,S_i^\alpha)$ of $S$. In order to obtain a well-posed system, one must specify the dependence of stress and force on the configuration $\vp$.

Back to the time-dependent context, of the force $f\in T_\kappa^*\QQ$ is of the form \eqref{timeindforce}, where $f_t\in T_{\kappa_t}^*\Q$, then there exists a family $\sigma_t$ of stresses at $\kappa_t$, such that
\[
f(w) = \int_I \sigma_t(j^1w_t)\, dt.
\]
If, furthermore, every $\sigma_t$ is smooth with a family of variational stress densities $S_t$, then
\[
f(w) = \int_I \brk{-\int_\B \divergence S_t(w_t) + \int_{\partial\B}p_\sigma S_t(w_t)}\, dt.
\]
The representation theorem states then that a force $f$, given by time-dependent body and surface force densities,
\[
f(w) = \int_I \brk{\int_\B b_t(w_t) + \int_{\partial\B} \calT_t(w_t)}\, dt,
\]
is represented by a family of smooth stresses with densities $S_t$, then
\[
\divergence S_t + b_t=0 
\Textand  
p_\sigma S_t|_{\partial\B}=\mathcal{T}_t.
\]

\subsection{Loadings and constitutive relations}

A mechanical system, whether finite- or infinite-dimensional, is specified by its configuration space, and by a force field, assigning a force to every configuration. It is customary in mechanics to partition the total force $F_T$ into external and internal components; in continuum mechanics external forces are due to \emph{loadings}, and internal forces result from a \emph{constitutive relation}.

In our setting, a force field is a 1-form on the configuration space, $F_T\in\Gamma(T^*\QQ)$. 
%n order to write the equations of motion, we need to know the dependence of both internal and external forces on the configuration. 
%
%\begin{definition}
%A \emph{loading} is a $1$-form $\Phi:\QQ\to T^*\QQ$, assigning to every configuration $\kappa\in\QQ\,$ a force $\Phi_\kappa\in T^*_\kappa\QQ$.
%A \emph{constitutive relation} is a section $\Psi:\QQ\to \s$, assigning to every  $\kappa\in\QQ$ a stress 
%\[
%\Psi_\kappa\in \s_\kappa\simeq (C^0(J^1(\ks)))^*.
%\]
%\Raz{(Don't you want to define it such that the stress at $\kappa$ is $(j^1\kappa)^*\Psi$?)}\Elik{agreed}
%The total force at a given configuration $\kappa$ is given by 
%\[
%(F_T)_\kappa=\Phi_\kappa-\Psi_\kappa\circ j^1.
%\]
%\end{definition}
We will focus our interest on time-independent force fields, i.e., force fields induced by section of $T^*\Q$. To this end, 
%The inclusion 
%\[
%\iota_\Q:\Q\hookrightarrow\QQ,\quad \iota_Q(f)(t,x)=f(x),
%\]
% induces the pullback of sections,
%\[
%\iota_\Q^\star:\Gamma(T^*\QQ)\to \Gamma(T^*\Q)
%\]
%defined by 
%\[
%(\iota_\Q^\star \Phi)_g(u)=\Phi_{\iota_\Q(g)}(T\iota_\Q(u)),\quad\forall g\in\Q,\, u\in T_g\Q.
%\]
define the extension map $\ext:\Gamma(T^*\Q)\to \Gamma(T^*\QQ)$,
\beq
\label{extension operator}
(\ext(F))_\kappa(w)=\frac{1}{|I|}\int_I F_{\kappa_t}(w_t)dt,\quad \kappa\in\QQ,\,w\in T_\kappa\QQ,
\eeq
where for every $t\in I$, $\kappa_t$ and $w_t$ were defined above. 
This extension is natural for the following reason: 
The inclusion $\iota_\Q$, defined by \eqref{eq:iotaQ}, induces a pullback of sections,
\[
\iota_\Q^\star:\Gamma(T^*\QQ)\to \Gamma(T^*\Q),
\]
defined by 
\[
(\iota_\Q^\star F)_\vp(u)=F_{\iota_\Q(\vp)}(T\iota_\Q(u)),\qquad  \vp\in\Q,\, u\in T_\vp\Q.
\]
A straightforward calculation shows that $\ext$ is a right-inverse for $\iota_\Q^\star$,
\beq\label{inclusion relation}
\iota_\Q^\star \circ \ext= Id_{\Gamma(T^*\Q)}.
\eeq
%On the other hand , let $F\in\Gamma(T^*\QQ)$, $\kappa\in\QQ$ and $v\in T_\kappa\QQ$, then 
%\[
%\begin{split}
%(\ext\circ \iota_\Q^\star)(F)_\kappa(v)=\int_I (\iota_\Q^\star(F))_{\kappa_t}(v_t)dt=\frac{1}{|I|}\int_IF_{\iota_\Q(\kappa_t)}(T\iota_\Q(v_t))dt.
%\end{split}
%\]

\begin{definition}
We say that a force field $F\in\Gamma(T^*\QQ)$ is \emph{time-independent} if $F=\ext(\Phi)$ for some $\Phi\in \Gamma(T^*\Q)$.
\end{definition}

Thus, a time-independent force field depends on time only through the time dependence of the configuration; moreover, by definition, its action on a virtual displacement $w\in T_\kappa\QQ$ does not involve explicitly the time derivative of $w$.

With the aid of the extension operator, we now show how the total force is composed from a loading and a constitutive relation:

\begin{definition}
A \emph{loading} is a $1$-form $\Phi:\Q\to T^*\Q$, assigning to every $\vp\in\Q\,$ a force $\Phi_\vp\in T^*_f\Q$.
A \emph{constitutive relation} is a section $\Psi:C^0(J^1(\B,\S))\to T^*C^0(J^1(\B,\S))$; its induced force field $(j^1)^*\Psi:\Q\to T^*\Q$ is given by
\[
((j^1)^*\Psi)_\vp(w)=\Psi_{j^1\vp}(j^1w).
% (C^0(J^1(f^*T\S)))^*.
\]
The total force at a given configuration $\kappa\in\QQ$ is given by 
\[
F_T = \ext(\Phi-(j^1)^*\Psi)\in\Gamma(T^*\QQ).
\]
\end{definition}

That is, for $\kappa\in\Q$ and $w\in T_\kappa\QQ\simeq C^1(\ks)$
\[
(F_T)_\kappa(w)= \frac{1}{|I|} \int_I \left(\Phi_{\kappa_t}(w_t)-\Psi_{j^1\kappa_t}(j^1w_t)\right)dt.
\]
Note that by the isomorphism \eqref{eq:iso5}, 
\[
((j^1)^*\Psi)_\vp\in T^*_{j^1\vp}C^0(J^1(\B,\S))\simeq (C^0(J^1(\vp^*T\S)))^*.
\]

We next restrict our attention to smooth loading and smooth constitutive relations, which are induced by densities in the form of sections of vector bundles over $\B\times\S$.
%
%Given two manifolds $M$ and $N$ denote their  product by $\underline{N}_M:=M\times N$. $\underline{N}_M$ is naturally a fiber bundle over $M$, moreover, if $\pi :N\to U$ is a fiber bundle then $\underline{N}_M$ is a fiber bundle over $M\times U$. 
%\Elik{Maybe move to introduction. Is there a better way to define the following?}

\begin{definition}
A loading $\Phi\in\Gamma(T^*\Q)$ is called \emph{smooth} if there exists a \emph{body loading density}
\[
b\in \Gamma(\Hom(\B\times T\S,\Lambda^dT^*\B\times\S)),
\] 
and a \emph{surface loading density} 
\[
\mathcal{T}\in\Gamma(\Hom(\partial\B\times T\S,\Lambda^{d-1}T^*\partial\B\times\S)),
\]
such that for every $\vp\in\Q$ and $v\in T_\kappa\Q\simeq C^1(\vp^*T\S)$
\[
\Phi_\vp(v)=\int_\B \vp^*b(v)+\int_{\partial\B}\vp^*\mathcal{T}(v).
\]
\end{definition}

Note that
\[
\vp^*b \in \Gamma(\Hom(\vp^*T\S,\Lambda^dT^*\B)),
\]
so that $\vp^*b(v)$ is a $d$-form on $\B$, as required.

%%%%%%%%%%
\begin{definition}
A constitutive relation $\Psi$ is called \emph{smooth} if there exists a \emph{constitutive density},
\[
\psi\in \Gamma(\Hom(VJ^1(\B,\S),(\pi^1)^*\Lambda^dT^*\B)),
\]
such that for every $\vp\in\Q$ and $v\in T_\vp\Q$
\beq
\label{smooth stress}
\Psi_{j^1\vp}(j^1v)=\int_{ \B} ((j^1\vp)^*\psi)(j^1v).
\eeq
\end{definition}

Note that in \eqref{smooth stress} we used the canonical isomorphism \cite{KOS17b},
\[
(j^1\vp)^*VJ^1(M,N)\simeq J^1(\vp^*TN).
\]
%\Elik{different approach:  The projection  $\tau_2:I\times \B\to \B$ induces an inclusion
%\[
%\mathcal{I}:C^0(J^1(\B,\S))\hookrightarrow C^0(J^1(I\times \B,\S)),\quad j^1f\mapsto j^1(f\circ\tau_2)
%\]
%We require $\psi=\tilde\psi\circ\delta\mathcal{I}$ where $\tilde{\psi}$ is time-independent. 
%}
%%%%%%
%%%%%%%
%%%%%

%%%%%%%%%%%%%%%%%%%%%%
\subsection{The equations of motion}

In this section we establish the equations of motion as a generalization of Newton's second law of classical mechanics. We view the equations as a section of $T^*\QQ$, thus, velocity, momentum and acceleration are defined as sections of $T\QQ$ or $T^*\QQ$. 

The \emph{velocity} $V\in \Gamma(T\QQ)$ is defined by 
 \[
 V_\kappa=\pd{\kappa}{t}:=T\kappa(\partial_t)\in T_\kappa\QQ.
 \]
% The tangent mapping of $V$, 
% \[
% TV:T\QQ\simeq C^2(I\times \B, T\S)\to T^2\QQ\simeq C^2(I\times\B, T^2\S)
% \]
%Let $\kappa\in\QQ$ then 
% \[
% (TV)_\kappa:T_\kappa\QQ\simeq C^2(\kappa^*T\S)\to T^2_{V_\kappa}\QQ\simeq C^2(V_\kappa^*T^2\S)
% \]
The tangent map of $V$, 
\[
TV:T\QQ\to T^2\QQ
\]
can be computed explicitly. Let $\kappa\in\QQ$ and let $w\in T_\kappa\QQ$ be represented by a path $\gamma:(-\epsilon,\epsilon)\to \QQ$ satisfying $\gamma(0)=\kappa$ and $\dot{\gamma}(s)=w$. Then,  
 \[
 (TV)_\kappa(w)=\ds V_{\gamma(s)}=\ds \brk{\deriv{\gamma_s}{t}} =\deriv{}{t}\left(\left.\frac{d\gamma_s}{ds}\right|_{s=0}\right)=\deriv{w}{t}=Tw(\partial_t).
 \]
Note that we view $w$ as an element of $C^2(I\times\B,T\S)$ (with $\pi_\S(w)=\kappa$), hence, $dw/dt:I\times \B\to T^2\S$; moreover, $\pi_{T\S}\circ dw/dt=V_\kappa$. In other words, $dw/dt$ is an element of $C^1(V_\kappa^*T^2\S)\hookrightarrow C^1(I\times \B,T^2\S)$, consistent with the isomorphism
  \[
T^2_{V_\kappa}\QQ\simeq C^1(V_\kappa^*T^2\S).
 \]
 
Next, define the \emph{acceleration} $A\in\Gamma(T\QQ)$ by 
\[
A=\nabla^\QQ_VV=K^\QQ\circ TV(V).
\]
Let $\kappa\in\QQ$. Then, $A_\kappa\in C^1(\ks)$ is given by 
\beq
\label{Acc}
\begin{split}
A_\kappa&=K^\Q(T(V)(V))_\kappa=K^\S\circ \left((TV)_\kappa(V_\kappa)\right)\\
&=(K^\S\circ TV_\kappa(\partial_t))= (\kappa^*\nabla^\S)_{\partial_t}V_\kappa:=\frac{DV_\kappa}{dt},
\end{split}
\eeq
where the second equality follows from the definition of $K^\Q$, the third equality follows from the expression for $TV$, and the fourth equality follows from the definition of the pullback connection.

%\Elik{verify last  equality , perhaps add a short note on pullbacks of connection map}.
 
% Equation \eqref{Acc} shows that this definition is compatible with our former definition for acceleration.
 The \emph{momentum} $P\in \Gamma(T^*\QQ)$  is the dual pairing of the velocity $V\in\Gamma(T\QQ)$ with respect to the metric $\GG$ defined in \eqref{metricforQ}, 
 \[
 P=\flat^\GG(V):=\GG(V,\cdot).
 \]
  For $\kappa\in\QQ$ and $w\in T_\kappa\QQ$,
 \[
 P(w)=\GG(V,w)=\int_{I\times\B}\kappa^*G(V_\kappa,w)\theta\wedge dt.
\]
% However, note that notation is a bit misleading as $P^{\mathtt{t}}$ is not an induced section on $T^*\QQ$ like $F^{\mathtt{t}}$ (the momentum isn't a section of $T^*\Q$). 
 % \Elik{TO Do: establish a relation between $V^\mathtt{t},P^\mathtt{t}$ and $V,P$ in terms of the differential of the evaluation map.}
The \emph{inertial force} $DP/dt \in\Gamma (T^*\QQ)$ is defined by
 \[
 \frac{DP}{dt}:=\nabla^{\QQ^*}_VP.
 \]
 where $\nabla^{\QQ^*}$ is the dual connection of $\nabla^\QQ$ for $T^*\QQ$:
given $\xi\in\Gamma(T\QQ)$, 
\[
\frac{DP}{dt}(\xi)=\left(\nabla^{\QQ^*}_VP\right)(\xi)=V\cdot P(\xi)-P(\nabla_V^\QQ\xi).
\]
 
\begin{proposition}
The inertial force is dual to the acceleration   
 \[
 \frac{DP}{dt}=\flat^{\GG}(A).
 \]
 \end{proposition}
 
 \begin{proof}
 Let $\xi\in \Gamma(T\QQ)$. By Theorem \ref{thm:metricity}, $\nabla^\QQ$ is metric with respect to $\GG$. Hence, 
  \[
 V\cdot P(\xi)=V\cdot \GG(V,\xi)=\GG(\nabla^\QQ_VV,\xi)+\GG(V,\nabla^\QQ_V\xi),
 \]
and
\[
\frac{DP}{dt}(\xi)=V\cdot P(\xi)-P(\nabla_V^\QQ\xi)=\GG(A,\xi)=\flat^\GG(A)(\xi). 
\]
 \end{proof}

The equations of motion equate the inertial force with the forces induced by loadings and constitutive relations,
\beq
\label{eq:motion}
\frac{DP}{dt}=\ext(\Phi-(j^1)^*\Psi). 
\eeq
It is an equation taking values in $T^*\QQ$; its solutions are configurations $\kappa\in\QQ$. Generally, \eqref{eq:motion} has to be augmented by initial conditions; boundary conditions are already incorporated in the loadings and the constitutive relations.

Loadings and  the constitutive relations may be singular, in which case \eqref{eq:motion} may not have a local differential form. If the loading and the constitutive relation are smooth,  then \eqref{eq:motion} at $\kappa$ transforms into 
\beq
\label{eq:motion1}
\begin{split}
\int_{I\times\B} \kappa^*G(A_\kappa,w)\theta\wedge dt=
\int_I\int_\B (\kappa_t^*b(w_t)+\divergence (j^1\kappa_t)^*\psi(w_t)) dt\\
+\int_I\int_{\partial\B} \brk{\kappa_t^*\mathcal{T}(w_t)-p_\sigma((j^1\kappa_t)^*\psi)(w_t)} dt, \quad \forall w\in T_\kappa\QQ.
\end{split}
\eeq
Since equation \eqref{eq:motion1} holds for every vector field $w$, we obtain the following differential system:
\beq
\label{eq:motion2}
\kappa^* G(A_\kappa,\cdot)\, \theta = \kappa^*b(w) +  \divergence (j^1\kappa)^*\psi,
\eeq
which is an identity of vector valued forms in $I\times\B$ together with boundary conditions 
\[
\mathcal{T}_\kappa=p_\sigma((j^1\kappa)^*\psi)|_{I\times\partial\B}. 
\]

%%%%%%%

A stationary configuration $\vp\in\Q$ is called  an \emph{equilibrium configuration} if 
\[
\Phi_\vp-((j^1)^*\Psi)_\vp=0
\]
or equivalently, if the constant motion $\iota_\Q(\vp)\in\QQ$ is a solution of \eqref{eq:motion}.
 In the smooth case, the equilibrium condition yields the boundary value problem
\[
\begin{split}
&\divergence (j^1\vp)^*\psi + \vp^*b  = 0\qquad \text{in}\,\B\\
&\mathcal{T}_\vp=p_\sigma((j^1\vp)^*\psi_\vp)\qquad\quad\,\,\,\,\text{on}\,\partial\B.
\end{split}
\]
\
\begin{remark}
The solution of the force free equation 
\[
\frac{DP}{dt}=\GG(A,\cdot)=0
\]
is a geodesic flow of $\B$ in $\S$. This is a covariant version of Newton's law of inertia for non-Euclidian continuum dynamics. 
\end{remark}

%%%%%%
\subsection{The hyperelastic case}
\label{sec:hyper}
\begin{definition}
A constitutive relation $\Psi$ is called \emph{conservative} if the exists a differentiable function $U:C^0(j^1(\B,\S))\to \R$ such that for every $\vp\in\Q$ 
\[
\Psi_{j^1\vp}=(dU)_{j^1\vp}\in T^*_{j^1\vp}C^0(J^1(\B,\S))\simeq C^0(J^1(\vp^*T\S))^*. 
\]
A  constitutive relation $\Psi$ is called \emph{hyperelastic} if $\Psi$  is conservative and $U$ is of the form 
\[
U(j^1\vp)=\int_{\B}\LL(j^1\vp) \theta,
\]
where $\LL\in C^\infty(J^1(\B,\S))$ is a \emph{Lagrangian density function}.
\end{definition}
\begin{proposition}
\label{prop:hyper}
Let  $\Psi$ be a hyperelastic constitutive relation with Lagrangian destiny $\LL$. Then $\Psi$ is smooth and the constitutive density $\psi$ is given by 
\[
\psi=\delta\LL\otimes\theta
\] 
where $\delta\LL$ is the fiber derivative of $\LL$, i.e., the restriction of $d\LL$ to $VJ^1(\B,\S)$. Thus, for every $\vp\in\Q$, $(j^1\vp)^*\psi=\delta_{j^1\vp}\LL\otimes\theta$ where $\delta_{j^1\vp}\LL:=(j^1\vp)^*\delta\LL\in \Gamma(J^1(\vp^*T\S)^*)$. 
\end{proposition}

For a proof see \cite{KOS17b}. 

Locally, $\LL$ is represented by a function $\R^m\times\R^{d\times m}\to\R$, and for every $w\in T_\vp\Q$
\[
(j^1\vp)^*\psi(w^i,w^i_{,\alpha}) = (R_i w^i+\psi^\alpha_i w^i_{,\alpha}) \, \Vol.
\]
where
\beq\label{hyperstressformula}
R_i = \rho\,\pd{\LL}{y^i}\,(j^1\vp)
\textand 
\psi^\alpha_i = \rho\,\pd{\LL}{y^i_{,\alpha}}(j^1\vp),
\eeq
and $\rho$ is the mass density.
In the absence of external loadings the equation  of motion \eqref{eq:motion2} take the form
%,  the equation of motion \eqref{eq:motion} for $\kappa\in\QQ$ takes the  form 
%\beq
%\label{hyperdyn}
%\int_{I\times\B} G(A_\kappa,w)\theta\wedge dt=
%\int_{I\times\B} \divergence S_{\kappa}(w) - \int_{I\times\partial\B}p_\sigma (S_\kappa)(w),\quad\forall w\in T_\kappa\QQ.
%\end{equation}
\begin{equation}\label{hyperdyncord}
G_{ij}\brk{\pdd{\kappa^i}{t}+\Gamma_{lk}^i\pd{\kappa^l}{t}\pd{\kappa^k}{t}} =
\frac{1}{\rho}\partial_\alpha\brk{\rho \pd{\LL}{y^j_{,\alpha}}(j^1\kappa)} - \pd{\LL}{y^j}(j^1\kappa),
\eeq
with boundary conditions
\[
\pd{\LL}{y^i_{,\alpha}}(j^1\kappa)\, (\partial_\alpha\inc\Vol) =0  
\qquad \text{on}\,I\times\partial\B.
\]
Eq. \eqref{hyperdyncord} is the equation of motion for the configuration $\kappa$ of a hyperelastic body in the absence of external loadings. It should supplemented by initial conditions $\kappa_0\in\Q$ and $V_0\in T_{\kappa_0}\Q$.

%\[
%\kappa^* G(A_\kappa,\cdot)\, \theta =  \divergence \delta_{j^1\vp}\LL\otimes\theta,
%\]

 %%%%%%%%%%%%%%
 %%%%%%%%%%%%%%%%%%
 \section{Linearization}
\label{sec:linearization}

We begin by defining the notion of linearization in a general context: 

\begin{definition}
\label{def:linearization}
Let $\pi: E\to M$ be a (possibly infinite dimensional) vector bundle, $\nabla$ a connection on $E$  and $s\in\Gamma(E)$. The \emph{linearization} of $p$ at a point $p\in M$, denoted by $L_{p}s\in\Aff(T_pM,E_p)$, is defined by
\[
L_ps(v):=s_p+ (\nabla_v s)_p\qquad  v\in T_pM .
\]
\end{definition}

Linearizations are used, in particular, in the following context: one seeks a solution $p\in M$ to the (generally nonlinear) equation
\[
s_p = 0.
\]
Instead, one takes an approximate solution $p_0$, and then solves the linear equation
\[
L_ps(v) = 0.
\]
Then, $p_1 = \exp_{p_0}(v)$ can be viewed as a next order iterate for the solution.

In our setting, $E=T^*\QQ$, $M=\QQ$ and the section $s\in \Gamma(T^*\QQ)$ is given by the equations of motion \eqref{eq:motion}
\[
s =\frac{DP}{dt} - \ext(\Phi- (j^1)^*\Psi) =0.
\]
In this case, $L_\kappa s \in \Aff(T_\kappa\QQ,T^*_\kappa\QQ)$. 
The linearized equation of motion at $\kappa\in\Q$ for $w\in T_\kappa\QQ$ is  
\beq
\label{linearizedeqofmotion}
L_\kappa\brk{\frac{DP}{dt}}(w) = L_\kappa(\ext(\Phi- (j^1)^*\Psi))(w).
\eeq
A solution $w\in T_\kappa\QQ$ for \eqref{linearizedeqofmotion} induces an approximate solution  $\kappa_1=\phi_\kappa(w)\in\QQ$ to \eqref{eq:motion},  where $\phi_\kappa$ is a canonical chart at $\kappa$. 

Note that a solution $w$ of \eqref{linearizedeqofmotion} saitsfies
\beq
\label{linearizedeqofmotio2}
L_\kappa\brk{\frac{DP}{dt}}(w)(\xi) = L_\kappa(\ext(\Phi- (j^1)^*\Psi))(w)(\xi),
\eeq
for every $\xi\in T_\kappa\QQ$. In order to compute \eqref{linearizedeqofmotio2} explicitly, one needs to consider a local extension of $\xi$, that is a local section $\tilde{\xi}\in\Gamma(T\QQ)$ satisfying $\tilde{\xi}_\kappa=\xi$ and the same value is obtained regardless of how $\xi$ is extended in a vicinity of $\kappa$. Noting that $\xi\in C^1(\ks)$ is a vector field along $\kappa$, we may extend $\xi$ to a vector field on $\S$, $\hxi\in \Gamma(T\S)$. In particular, $\xi=\kappa^*\xi^\S$. Thus, it suffices  to impose that \eqref{linearizedeqofmotio2} be satisfied for $\xi$ of the form
\[
\xi: \kappa \mapsto \kappa^*\hxi,
\]
where $\hxi\in \Gamma(T\S)$.

%%%%%%
\subsection{Linearization of acceleration term}

%\Elik{Add remark about connection of solutions of the linearization to solutions of the original equation. Perhaps discuss applications of inverse function theorem}

Let $\kappa\in\QQ$ and  $w\in T_\kappa\QQ$, then
\[
L_\kappa\brk{\frac{DP}{dt}}(w)=L_\kappa(\flat^{\GG}(A))(w)=\flat^{\GG}(A_\kappa)+\nabla^{\QQ^*}_w\left(\flat^{\GG}(A)\right).
\]
We therefore turn to compute $\nabla^{\QQ^*}_w\left(\flat^{\GG}(A)\right)$. Let  $\xi\in \Gamma(T\QQ)$, then by the metricity of $\GG$,
\[
\begin{split}
\nabla^{\QQ^*}_w\left(\flat^{\GG}(A)\right)(\xi) &= \left(w\cdot \flat^{\GG}(A)(\xi)\right)_\kappa-\flat^{\GG}(A)(\nabla^\QQ_w\xi) \\
&=\left(w\cdot \GG(A,\xi)\right)_\kappa-\GG(A,\nabla^\QQ_w\xi)_\kappa \\
&=\GG(\nabla^\QQ_wA,\xi)_\kappa.
\end{split}
\]
In other words, metricity implies that 
\[
\nabla^{\QQ^*}_w\left(\flat^{\GG}(A)\right)=\flat^{\GG}(\nabla^\QQ_wA).
\]
It remains to compute $\nabla_w^\QQ A$. 
%Recall that (see manifolds of mapping notes) given a manifold $M$ (in our case $\QQ$ or $\S$) and connection map $K:T^2M\to TM$ we have an induced connection map for $TM$, $K_T:T^3M\to T^2M$ fully defined\footnote{As the mapping 
%\[(\tau_{TM},T\tau_M,K):T^2M\to TM\times (TM\oplus TM)
%\] is an isomorphism.} by the following properties:
%\begin{enumerate}
%\item $\tau_{TM}\circ K_T=\tau_{TM}\circ\tau_{T^2M}.$
%\item $T\tau_M\circ K_T=K\circ T^2\tau_M$.
%\item $K\circ K_T=K\circ TK-R(T\tau_M\circ T^2\tau_M,\tau_{TM}\circ T\tau_{TM},\tau_{TM}\circ T^2\tau_M).$
%\end{enumerate}
%Where $R$ is the curvature tensor
%\[
%R(X,Y,Z)=\nabla^2Z(X,Y)-\nabla^2Z(Y,X).
%\]
%Taking $M=\QQ$ and denoting  $\mathcal{V}:=T(V)(V)$ we obtain
%\[
%\begin{split}
%\nabla_w^\QQ A&=K^\QQ\circ TA_\kappa(w)=K^\QQ\circ T\left(K^\QQ\circ \mathcal{V}  \right)_\kappa(w)=\\
%&=K^\QQ\circ\left( (TK^\QQ)_{\mathcal{V}_\kappa}\circ T\mathcal{V}_\kappa \right)(w)=\\
%&=K^\QQ\circ K^\QQ_T\circ  T\mathcal{V}_\kappa\circ w+\\
%&+R(T\tau_\Q\circ T^2\tau_\QQ (T\mathcal{V}_\kappa(w)),\tau_{T\QQ}\circ T\tau_{T\QQ}(T\mathcal{V}_\kappa(w)),\tau_{T\QQ}\circ T^2\tau_\QQ(T\mathcal{V}_\kappa(w))).
%\end{split}
%\]
%\Elik{This won't give us much because we don't understand the curvature tensor of $\QQ$ perhaps do the corresponding calculation with $S$ (as $K^\QQ$ is given by composition).}
%
%A different approach: 
Let $\gamma:(-\epsilon,\epsilon)\to\QQ$ be a curve representing $w\in T_\kappa\QQ$, that is $\gamma(0)=\kappa$ and $\dot{\gamma}(0)=w$. Then,
\[
\nabla_w^\QQ A=K^\QQ\circ TA_\kappa(w)=K^\S\circ \ds A_{\gamma(s)}.
\]
Hence by the definition of the pullback connection $\gamma^*\nabla^\S$,
\[
\begin{split}
\nabla_w^\QQ A&=\ds\frac{D}{dt}\frac{d{\gamma}}{dt}=\\
&=\frac{D^2}{dt^2}\brk{\left.\frac{d\gamma}{ds}\right|_{s=0}} -
R^\S\brk{\left.\frac{d\gamma}{dt}\right|_{s=0},\left.\frac{d\gamma}{ds}\right|_{s=0},
\left.\frac{d\gamma}{dt}\right|_{s=0}}= \\
&=\frac{D^2w}{dt^2}+R^\S(w,V_\kappa,V_\kappa),
\end{split}
\]
where $R^\S$ is the curvature tensor corresponding to $\nabla^\S$,
\[
R^\S(X,Y,Z)=\nabla^\S_X\nabla^\S_YZ-\nabla^\S_Y\nabla^\S_XZ-\nabla^\S_{[X,Y]}Z, \quad X,Y,Z\in\Gamma(T\S).
\]
To conclude, for every $\kappa\in\QQ$ and $w\in T_\kappa\QQ\simeq C^1(\ks)$, 
\beq
\label{dyn-lin}
L_\kappa\brk{\frac{DP}{dt}} (w)=\flat^\GG(A_\kappa)+\flat^\GG\brk{\frac{D^2w}{dt^2}+R^\S(w,V_\kappa,V_\kappa)}.
\eeq
In other words, the linearization of acceleration term is the Jacobi equation. 
%Note that usually the Jacobi equation is a linearization of the geodesic equation at a geodesic; in the present context, the linearization is at $\kappa(\cdot
%Unlike the standard the standard Jacobi equation for every $x\in\B$ the equation is along the curve $\kappa(\cdot,x)$ which in general is not a geodesic. 

%%%%%
\subsection{Linearization of force}

We now turn to linearize the right-hand side of the equations of motion \eqref{eq:motion}. Without loss of generality, we may assume that there are no external loadings, $\Phi=0$, as the loading may be incorporated into the constitutive relation $\Psi$. Thus the total force is given by,
\beq
F_\kappa(u) = \ext((j^1)^*\Psi)_\kappa(u)
=\frac{1}{|I|}\int_I \Psi_{j^1\kappa_t}(j^1u_t)\, dt,\quad \kappa\in\QQ,\,u\in T_\kappa\QQ.
\label{eq:yet_another_F}
\eeq
Then,
\[
L_\kappa(F)(w)=F_\kappa+ (\nabla^{\QQ^*}_wF\left.\right)_\kappa,
\]
where by definition, for $\xi\in\Gamma(T\QQ)$,
\[
(\nabla^{\QQ^*}_wF)(\xi)= w(F(\xi)) -F(\nabla^\QQ_w\xi).
\]

%%%%%%%%%%%%%
\begin{lemma}
%\label{lem:}
Let $F\in\Gamma(T^*\QQ)$ be given by \eqref{eq:yet_another_F}. Then, for every vector field $\xi\in\Gamma(T\QQ)$, 
\beq
\label{lin:12}
L_\kappa(F)(w)(\xi)=\frac{1}{|I|}\int_I L^\Q_{\kappa_t}((j^1)^*\Psi)(w_t)((\xi_{\kappa})_t)\,dt,
\eeq
where  the linearization on the right-hand side takes place in the space of stationary configurations $\Q$.
\end{lemma}
%%%%%%%%%%%%%

%%%%%%%%%%%%%
\begin{proof}
The constant part of the identity is immediate since $F$ is the extension of $(j^1)^\star\Psi$.  
To proceed as noted above, it suffices to consider vector field $\xi$  of the form $\kappa\mapsto \kappa^*\hxi$, where $\hxi\in \Gamma(T\S)$. Note also that the mapping $\vp\mapsto \vp^*\hxi$ for $\vp\in\Q$ is a section of $T\Q$, which we denote by $\xiQ$. Then, $(\xi_\kappa)_t = \xiQ_{\kappa_t}$.

It remains to show the identity of the linear parts: that for every $\hxi\in\Gamma(T\S)$,
\beq
w(F(\xi)) = \frac{1}{|I|}\int_I w_t\brk{(j^1)^\star\Psi(\xiQ)}\, dt,
\label{eq:show1}
\eeq
and
\beq
F(\nabla^\QQ_w\xi) = \frac{1}{|I|}\int_I(j^1)^\star\Psi(\nabla^\Q_{w_t}\xiQ)\,dt.
\label{eq:show2}
\eeq

To show \eqref{eq:show1}, let $\gamma:(-\e,\e)\to \QQ$ satisfy $\gamma(0)=\kappa$ and $\dot{\gamma}(0)=w$, and let $\gamma_t:(-\e,\e)\to\Q$ be the evaluation of $\gamma$ at time $t$, so that  $\dot{\gamma}_t(0) = w_t\in T_{\kappa_t}\Q$. Then,
\[
\begin{split}
w(F(\xi)) &= \ds \left(\frac{1}{|I|}\int_I\Psi_{j^1\gamma_t(s)}(j^1\xiQ_{\gamma_t(s)})\, dt\right)\\
&=  \left(\frac{1}{|I|}\int_I \ds\Psi_{j^1\gamma_t(s)}(j^1\gamma_t(s)^* \hxi)\, dt\right) \\
&= \frac{1}{|I|}\int_I w_t\brk{(j^1)^\star\Psi(\xiQ)}.
\end{split}
\]

To show \eqref{eq:show2}, 
we first simplify the term $\nabla^\QQ_w\xi$. By the chain rule,
\[
(T\xi)_\kappa(w)= \ds \xi_{\gamma(s)} = \ds \gamma(s)^*\hxi = \kappa^*T\hxi\circ w,
\]
which is an identity in $T^2_{\xi_\kappa}\QQ \simeq C^1(\xi_\kappa^*T^2\S)$.
Hence,
\beq
\nabla^\QQ_w\xi = K^\QQ( (T\xi)_\kappa(w))=K^\S\circ   \kappa^*T\hxi\circ w = \kappa^*(\nabla^\S\hxi)(w),
\label{eq:will_be_needed}
\eeq
where $\kappa^*(\nabla^\S\hxi)\in\Gamma(\Hom(\ks,\ks))$ is the pullback of $\nabla^\S\hxi\in \Gamma(\Hom(T\S,T\S))$. For $t\in I$,
\[
(\nabla^\QQ_w\xi )_t = \kappa_t^*(\nabla^\S\hxi)(w_t)  =\nabla^\Q_{w_t}\xiQ,
\]
where the last equality follows from the calculation yielding \eqref{eq:will_be_needed} over $\B$, rather than $I\times\B$. Then,
\[
\begin{split}
F(\nabla^\QQ_w\xi) &= \frac{1}{|I|}\int_I(j^1)^\star\Psi((\nabla^\QQ_w\xi)_t) \, dt\\
&= \frac{1}{|I|}\int_I(j^1)^\star\Psi(\nabla^\Q_{w_t}\xiQ)\,dt,
\end{split}
\]
which concludes the proof.
\end{proof}
%%%%%%%%%%%%%
%
%
% Then for every $\kappa\in\QQ$ and $w\in T_\kappa\QQ$
%Let $X\in\Gamma(T\QQ)$ and  $t\in I$. Denote $X_{\kappa,t}:=(T\ev_t)_\kappa(X_\kappa)$ where $\ev_t:\QQ\to\Q$ is given by $\ev_t(\kappa)=\kappa_t$. In view of the isomorphism $T\QQ\simeq C^1(I\times\B,T\S)$ we have $X_{\kappa,t}=X_\kappa(t,\cdot)\in C^1(\kappa_t^*T\S)$. 
%
%By definition,
%%We first turn to compute the term $\nabla^\QQ_wX$. Suppose that $X$ is a constant vector filed that is, for every $\kappa\in\QQ$, $X_\kappa=\kappa^*\hxi$ where $\hxi\in \Gamma(T\S)$. 
%%
%Assume that $X\in\Gamma(T\QQ)$ is a "constant" vector field of the form
%\beq
%\label{lin:14}
%X_\kappa=\kappa^*\hxi
%\eeq
%where $\hxi\in\Gamma(T\S)$, hence $X_{\kappa,t}=\kappa_t^*\hxi$ and denote by $\xi\in \Gamma(T\Q)$ the restriction of $X$ to $\Q$.
%We first show that 
%\beq
%\label{lin:12}
%L_\kappa(F)(w)(X)=\frac{1}{I}\int_I L^\Q_{\kappa_t}((j^1)^*\Psi)(w_t)(\xi_{\kappa_t})dt,
%\eeq
%where in the right hand side of \eqref{lin:12} the linearization takes place in the space of time independent configuration $\Q$. 
%

For every $\hxi\in\Gamma(T\S)$ and $\epsilon>0$ sufficiently small, consider the vector field $\xi_\e\in \Gamma(T\QQ)$ given by $\kappa\mapsto \chi_\epsilon(\kappa^*\hxi)$, where $\chi_\epsilon:I\times\B\to \R$ is a smooth cutoff function supported on $(-\epsilon,\epsilon)\times\B$. By evaluating \eqref{lin:12} and \eqref{dyn-lin} at $\xi_\e$ and letting $\e\to 0$ the linearized equations of motion \eqref{linearizedeqofmotio2} can be localized in time:

%%%%%%%%%%%%%
\begin{corollary}
\label{corr:local_in_time}
For a time-independent force induced by a constitutive relation $\Psi$ and zero loading, the linearized equations of motion \eqref{linearizedeqofmotio2} is local in time; $w\in T_\kappa\QQ$ solves \eqref{linearizedeqofmotio2} if and only if, for every $t\in I$,
\[
\int_\B \G\brk{(A_\kappa)_t + \frac{D^2w}{dt^2}+R^\S(w_t,(V_\kappa)_t,(V_\kappa)_t),\cdot}\theta  =  L^\Q_{\kappa_t}((j^1)^*\Psi)(w_t)
\]
which is an equality of co-vectors in $T^*_{\kappa_t}\Q$. 
\end{corollary}
%%%%%%%%%%%%%

%%%%%%%%%%%%%
%\begin{proof}
%Suppose that $w\in T_\kappa\QQ$ solves \eqref{linearizedeqofmotio2} and let $t\in I$
%\end{proof}
%%%%%%%%%%%%%

In view of \eqref{lin:12}, we need to calculate linearizations of the form
\[
L_\vp((j^1)^\star\Psi)(v),
\]
where $\vp\in\Q$ and $v\in T_\vp\Q$. We focus on the case where $\Psi$ is smooth, given by the constitutive density
\[
\psi\in  \Gamma(\Hom(VJ^1(\B,\S),(\pi_1)^*\Lambda^dT^*\B)).
\]
For every vector field $\xiQ\in\Gamma(T\Q)$,
\[
\begin{split}
L_\vp((j^1)^\star\Psi)(v)(\xiQ) &= 
\int_\B\psi_{j^1\vp}(j^1\xiQ_\vp)+v \left(\int_\B\psi_{j^1\vp}(j^1\xiQ_{\vp})\right)
-\int_\B\psi_{j^1\vp}(j^1(\nabla^\Q_v\xiQ)_\vp)\\
&= - \int_\B \divergence(\psi_{j^1\vp})\left(\xiQ_{\vp}\right)
+ \int_{\partial\B}p_\sigma(\psi_{j^1\vp})\left(\xiQ_{\vp}\right)
 \\
&\quad+\int_\B \divergence(\psi_{j^1\vp})\left((\nabla^\Q_v\xiQ)_{\vp}\right)
-v \left(\int_\B\divergence(\psi_{j^1\vp})(\xiQ_{\vp})\right)  \\
&\quad-\int_{\partial\B}p_\sigma(\psi_{j^1\vp})\left((\nabla^\Q_v\xiQ)_{\vp}\right) +
v\left(\int_{\partial\B}p_\sigma(\psi_{j^1\vp})(\xiQ_{\vp})\right),
\end{split}
\]
where we substituted the decomposition \eqref{stressrep} of $\psi$ into a divergence term and a boundary term.

To further simply the last equation, we note that for $\xiQ$ of the form $\vp\mapsto \vp^*\hxi$, in which case
\[
(\nabla^\Q_v\xi)_\vp=(\vp^*\nabla^\S\hxi)(v).
\]
Moreover, the equation is tensorial in $\xiQ$, so that it can be represented as
\beq
\label{lin:11}
L^\Q_\vp((j^1)^*\Psi)(v)(\xiQ) = ((j^1)^*\Psi)_\vp(\xiQ)  +
\int_\B \calA(\vp,v)(\xiQ)+\int_{\partial\B}\calB(\vp,v)(\xiQ),
\eeq
where $\calA(\vp,v)$ is a $d$-form valued in $(\kappa^*T\S)^*$ and $\calB(\vp,v)$ is a $(d-1)$-form valued in the same vector bundle.
At this stage and generality, $\calA(\vp,v)$ and $\calB(\vp,v)$ cannot be significantly simplified. We therefore turn to calculate their local representatives in a given coordinate chart.

\subsection{Local form of the linearized equations of motion}

Substituting \eqref{lin:11} into \corrref{corr:local_in_time}, we obtain linearized equations of the motion in local form,
\beq
\label{linearization_localform}
\flat^{\kappa^*\G}\brk{A_\kappa + \frac{D^2w}{dt^2}+R^\S(w,V_\kappa,V_\kappa)}\otimes\theta  = 
- \divergence(\psi_{j^1\kappa}) + \calA(\kappa,w)
\eeq
in $I\times\B$ and
\[
p_\sigma(\psi_{j^1\kappa}) +  \calB(\kappa,w) = 0
\]
in $I\times\partial\B$.

If $\kappa$ is a solution of \eqref{eq:motion}, then only the terms that are linear in $w$ remain. 
In the particular case where $\kappa = \iota_\Q(\vp)$ is a stationary solution of \eqref{eq:motion}, $V_\kappa=0$, hence 
\[
\flat^{\kappa^*\G} \brk{\frac{D^2w}{dt^2}} \otimes\theta  = \calA(\kappa,w) 
\qquad\text{in $I\times\B$},
\]
and
\[
\calB(\kappa,w) = 0 \qquad\text{in $I\times\partial\B$}.
\]
Moreover, since $\divergence(\psi_{j^1\vp})=0$ and $p_\sigma(\psi_{j^1\vp})=0$, the implicit expressions for $\calA$ and $\calB$ reduce to
\[
\begin{split}
\int_\B \calA(\vp,w)(\xiQ) &= 
-v \left(\int_\B\divergence(\psi_{j^1\vp})(\xiQ_{\vp})\right)  \\
\int_\B \calB(\vp,w)(\xiQ) &= v\left(\int_{\partial\B}p_\sigma(\psi_{j^1\vp})(\xiQ_{\vp})\right).
\end{split}
\]

%%%%%%%%%
\subsection{Coordinate representation}
%%%%%%%%%%%

We hereby give a local expression for the terms $\calA(\kappa,w)$ and $\calB(\kappa,w)$ in \eqref{linearization_localform} for the general case. 
For $(t,p)\in I\times\B$ let $x^\alpha:U_p\subset\B\to\R$ ($1\leq\alpha\leq d$) and $y^i:V_{\kappa(t,p)}\subset \S\to\R$ ($1\leq i\leq m$) be coordinate charts for $\B$ and $\S$ at $p$ and $\kappa(t,p)$ respectively such that $\kappa(I\times U_p)\subset V_{\kappa(t,p)}$. Then $\kappa\left.\right|_{I\times U_p}$ is represented by 
\[
\kappa^i=y^i\circ\kappa\circ (\id ,x^{-1}):I\times \R^d\to \R.
\]
$w\in C^1(\ks)$ then has the local form 
\[
w=w^i\, \kappa^*\partial_{y^i}
\]
where $w^i:I\times\B\to\R$ and $\{\partial_{y^i}\}_{i=1}^m$ is the local frame for $T\S$ induces by the charts $y^i$. 
With a slight abuse of notation, let
\[
(x^\alpha,y^i,A^i_\alpha):J^1(U_p,V_{\vp(p)})\subset J^1(\B,\S)\to\R^d\times \R^m\times \R^{d\times m}
\]
be the induced coordinate chart for $J^1(\B,\S)$. That is, for $j^1_qf\in J^1(U_p,V_{\vp(p)})$, 
\[
 x^\alpha(j^1_qf)=x^\alpha(q),\quad 
 y^i(j^1_qf)=y^i(f(q))\textand A^i_\alpha(j^1_qf)=\pd{f^i}{x^\alpha}(x(q)).
 \] 

The variational stress density $\psi\in\Gamma(L(VJ^1(\B,\S),\Lambda^dT^*\B)$ has the local form 
\[
\psi=(\psi^\alpha_i\,dA^i_\alpha+R_j\, dy^j)\Vol
\]
where $\Vol=dx^1\wedge\cdot\wedge dx^d$ and $\psi^\alpha_i,R_j:J^1(U_p,V_{\vp(p)})\to\R$.
hence, 
\[
\psi_{j^1\kappa_t}(j^1w_t)=\left(\psi^\alpha_i(j^1\kappa_t)\, w^i_{,\alpha}+
R_j(j^1\kappa_t)\, w^j   \right)\Vol.
\]
Finally, denote by $\Gamma^i_{jk}:V_{\vp(p)}\to\R$ the christoffel symbols of $\nabla^\S$.
 A straightforward calculation then gives:

Let $\Psi$ be a smooth constitutive relation represented by a constitutive density $\psi$ and let $\kappa\in\QQ$. Then the linearization of the equation of motion $\frac{DP}{dt}=\ext((j^1)^*\Psi)$, at $\kappa\in\QQ$, 
\[
L_\kappa\brk{\frac{DP}{dt}}(w)= L_\kappa(\ext((j^1)^*\Psi))(w),
\qquad w\in T_\kappa\QQ
\]
has the  local form,
\beq\label{eq9}
\begin{split}
\rho G_{ij} \left(A_\kappa^i+\frac{D^2w^i}{dt^2}+R^i_{hkl}\pd{\kappa^h}{t}\pd{\kappa^k}{t}\pd{w^l}{t}\right) &
= \calA^1(\kappa_t)_{ij}w^i+
\calA^2(\kappa_t)^\delta_{ij}w^i_{,\delta} + \calA^3(\kappa_t)^{\alpha\beta}_{lj}w^l_{,\alpha\beta}\\
&\hspace{-3cm}+ (\divergence\psi_{\kappa_t})_k(w^i \Gamma^k_{ij} - (\divergence\psi_{\kappa_t})_j),
\end{split}
\eeq 
in $I\times\B$, and
\[ 
\left(\psi^\alpha_i(j^1\left(1-w^l\Gamma^i_{lj}\right) + 
\pd{\psi^\alpha_j}{y^k} w^k+\pd{\psi^\alpha_j}{A^k_\beta} w^k_{,\beta}\right) (\iota_{\partial_\alpha}\Vol)\left.\right|_{\partial\B}
\]
on $I\times\partial\B$.
The function $\calA_1,\calA_2$ and $\calA_3$ are given by 
\[
\begin{split}
\calA^1(\vp)_{ij} &= \pddm{\psi^\alpha_j}{y^i}{x^\alpha}  + \pddm{\psi^\alpha_j}{y^i}{y^l}\vp^l_{,\alpha}+\pddm{\psi^\alpha_j}{y^i}{A^l_\beta}\vp^l_{,\alpha\beta}-\pd{R_j}{y^i},\\
\calA^2(\vp)^\delta_{ij}&=\pddm{\psi^\alpha_j}{A^i_\delta}{x^\alpha}+\pddm{\psi^\alpha_j}{A^i_\delta}{y^l}\vp^l_{,\alpha}  +  \pd{\psi^\delta_j}{y^i}   +   \pddm{\psi^\alpha_j}{A^i_\delta}{A^l_\beta}\vp^l_{,\alpha\beta}   -   \pd{R_j}{A^i_\delta},   
\end{split}
\]
and
\[
\calA^3(\vp)^{\alpha\beta}_{lj}=\pd{\psi^\alpha_j}{A^l_\beta}.
\]
In these equations, the entries $G_{ij}$, $R^i_{hkl}$ and $\Gamma^k_{ij}$ of the metric, the curvature and the connection are evaluated at $\kappa_t$; the entries $\psi_j^\alpha$, $R_j$ of the constitutive density and their derivatives are evaluated at $j^1\kappa_t$. 

%In case that the constitutive relation $\Psi$ is hyperlastic (see section \ref{sec:hyper}) with corresponding energy 
%\[
%U(j^1\vp)=\int_\B\LL(j^1\vp)\theta,
%\]
%the constitutive density (see proposition \ref{prop:hyper}) is given by $\psi=\delta\LL\otimes\theta$ and is locally expressed by 
%\[
%\psi^\alpha_j=\pd{\LL}{A^j_\alpha},\quad R_j=\pd{\LL}{y^j}.
%\]
% The leading term in the right hand side of \eqref{eq9}, sometimes denoted as the \emph{elasticity tensor},  then takes the form
% \[
% \calA^3(\vp)_{lj}^{\alpha\beta}=\left.\pddm{\LL}{A^i_\alpha}{A^j_\beta}\right|_{j^1\vp}.
% \]
%  If $\vp$ is a \emph{normal state}, that is $\delta_{j^1\vp}\LL\equiv 0$ then the linearization around $\vp$ takes the simple form 
% \[
% \rho G_{ij} \frac{D^2w^i}{dt^2}=\pddm{\LL}{A^i_\alpha}{A^j_\beta}w^i_{,\alpha\beta}
% \]  
% with boundary conditions
% \[
% ss
% \]

\paragraph{Acknowledgments}
This research was partially funded by the Israel Science Foundation (Grant No. 1035/17), and by a grant from the Ministry of Science, Technology and Space, Israel and the Russian Foundation for Basic Research, the Russian Federation.

 \bibliographystyle{amsalpha}
%\bibliography{/Users/raz/Dropbox/tex/Refs/MyBibs}
\bibliography{/Users/elikolami/Desktop/Latex/MyBibs}

\end{document}